\documentclass[twoside,leqno,twocolumn]{article}
\usepackage{ltexpprt}

\usepackage[]{graphicx}\usepackage[]{color}
\makeatletter
\def\maxwidth{ %
  \ifdim\Gin@nat@width>\linewidth
    \linewidth
  \else
    \Gin@nat@width
  \fi
}
\makeatother

\definecolor{fgcolor}{rgb}{0.345, 0.345, 0.345}

\usepackage{framed}
\makeatletter
 {\par\unskip\endMakeFramed%
 \at@end@of@kframe}
\makeatother

\definecolor{shadecolor}{rgb}{.97, .97, .97}
\definecolor{messagecolor}{rgb}{0, 0, 0}
\definecolor{warningcolor}{rgb}{1, 0, 1}
\definecolor{errorcolor}{rgb}{1, 0, 0}
\newenvironment{knitrout}{}{} % an empty environment to be redefined in TeX

\usepackage{alltt}
\newcommand{\SweaveOpts}[1]{}  % do not interfere with LaTeX
\newcommand{\SweaveInput}[1]{} % because they are not real TeX commands
\newcommand{\Sexpr}[1]{}       % will only be parsed by R

\usepackage{xspace}
\newcommand{\eg}{e.\,g.\xspace}

\usepackage{amssymb, amsmath}
\usepackage{graphicx}
\usepackage[autolanguage]{numprint}
\usepackage{pifont} 
\usepackage{multirow}
\usepackage[font=footnotesize]{caption}
\usepackage{sidecap}
\usepackage{wrapfig}
\usepackage[normalem]{ulem}

\newcolumntype{C}[1]{>{\centering\let\newline\\\arraybackslash\hspace{0pt}}m{#1}}
\newcolumntype{R}[1]{>{\raggedleft\let\newline\\\arraybackslash}m{#1}}
\newcolumntype{L}[1]{>{\raggedright\let\newline\\\arraybackslash}m{#1}}

\usepackage{siunitx,array}
\usepackage[usenames,dvipsnames]{xcolor}

\usepackage[algo2e,ruled,vlined,linesnumbered]{algorithm2e}
\SetCommentSty{textit}
\DontPrintSemicolon
\IncMargin{-\parindent}
\SetAlCapHSkip{0pt}
\SetAlgoLined
\makeatletter
\newcommand{\removelatexerror}{\let\@latex@error\@gobble}
\makeatother

\SetKwProg{Function}{Function}{}{}
\SetKwFunction{Activate}{Activate}
\SetKwFunction{rate}{rate}
\SetKwFunction{coarsen}{coarsen}
\SetKwFunction{DeltaGainUpdate}{DeltaGainUpdate}
\DeclareMathOperator*{\argmax}{\arg\!\max}

\newtheorem{AdaptiveE}{Lemma}[section]

\usepackage{hyperref}

\newcommand{\Inc}{{++}}
\newcommand{\Dec}{{--}}

\newcommand{\Remi}[1]{\tcp*[f]{#1}}

\begin{document}
\npthousandsep{\ }
\npdecimalsign{.} % we want . not , in numbers

\title{\Large $k$-way Hypergraph Partitioning via $n$-Level Recursive Bisection}
\author{Sebastian Schlag\thanks{Karlsruhe Institute of Technology, Institute for Theoretical
Informatics, Algorithmics II. Email: \{sebastian.schlag, meyerhenke, sanders, christian.schulz\}@kit.edu,
tobias.heuer@gmx.net, vitali.henne@gmail.com} 
\and 
Vitali Henne$^*$
\and 
Tobias Heuer$^*$
\and 
Henning Meyerhenke$^*$
\and 
Peter Sanders$^*$
\and 
Christian Schulz$^*$ }
\date{}

\maketitle

\pagestyle{headings}
\pagenumbering{arabic}
\setcounter{page}{1}%Leave this line commented out.

\begin{abstract} \small\baselineskip=9pt
We develop a multilevel algorithm for hypergraph partitioning that contracts the vertices one at a time. Using several caching and lazy-evaluation techniques during coarsening and refinement, we reduce the running time by up to two-orders of magnitude compared to a naive $n$-level algorithm that would be adequate for ordinary graph partitioning. The overall performance is even better than the widely used hMetis hypergraph partitioner that uses a classical multilevel algorithm with few levels.
Aided by a portfolio-based approach to initial partitioning and adaptive budgeting of imbalance within recursive bipartitioning, we achieve very high quality.
We assembled a large benchmark set with 310 hypergraphs stemming from
application areas such VLSI, SAT solving, social networks, and scientific computing. 
We achieve significantly smaller cuts than hMetis and PaToH, while being faster than hMetis.
Considerably larger improvements are observed for some instance classes like social networks, for bipartitioning, and for
partitions with an allowed imbalance of 10\%.
The algorithm presented in this work forms the basis of our hypergraph
partitioning framework \emph{KaHyPar} (\textbf{Ka}rlsruhe \textbf{Hy}pergraph \textbf{Par}titioning).
\end{abstract}

\section{Introduction} \label{Introduction}
Hypergraphs are a generalization of graphs, where each (hyper)edge can connect more than two vertices.
The $k$-way hypergraph partitioning problem is the generalization of the well-known graph partitioning problem:
Partition the vertex set into $k$ disjoint blocks of bounded size
(at most $1+\varepsilon$ times the average block size), while minimizing the total cut size, i.e., the sum of the weights of those hyperedges that connect multiple blocks.
However, allowing hyperedges of arbitrary size makes the partitioning problem more difficult in practice~\cite{Application:DistributedDB,Heintz:2014}.

Hypergraph partitioning (HGP) has a wide range of applications. Two prominent areas are 
VLSI design and scientific computing (\eg accelerating sparse matrix-vector multiplications)~\cite{Papa2007}.
While the former is an example of a field where small optimizations can lead to significant savings, the latter
is an example where hypergraph-based modeling is more flexible than graph-based approaches~\cite{Heintz:2014,PaToH,Hendrickson20001519,10.1371/journal.pcbi.1000385}. HGP also finds application as a preprocessing step in SAT solving, where it is used to identify groups of connected variables~\cite{DBLP:journals/jucs/AloulMS04}.

Since hypergraph partitioning is NP-hard~\cite{Lengauer:1990} and since it is even NP-hard to find good approximate solutions for graphs~\cite{DBLP:journals/ipl/BuiJ92},
heuristic algorithms are used in practice. The most commonly used heuristic is the \emph{multilevel paradigm}~\cite{MultiLevel_Bui,MultiLevel_Cong,MultiLevel_Hauck,MultiLevel_Hendrickson}. It consists
of three phases: In the \emph{coarsening phase}, the hypergraph is recursively coarsened to obtain a hierarchy of smaller hypergraphs that reflect the basic structure of the input. After applying an \emph{initial partitioning} algorithm to the smallest hypergraph in the second phase, coarsening is 
undone and, at each level, a \emph{local search} method is used to improve the partition induced by the coarser level.

State-of-the-art hypergraph partitioners use matching- or clustering-based algorithms to find 
groups of highly-connected vertices that can be contracted together to create the next level of the 
coarsening hierarchy~\cite{PaToH,hMetisRB,trifunovic2006parallel}. The rate at which successively coarser hypergraphs are reduced determines 
the number of levels in the multilevel hierarchy. As already noted in \cite{MLPart}, a larger number of levels potentially improves the solution quality, 
because local search algorithms are used more often. However, it also leads to larger running times and increased memory usage.

In this paper we explore the idea to evade this trade-off by going to the extreme case of (nearly) $n$ levels and removing
only a single vertex between two levels. 
Implementing this in a sophisticated way, we are able to combine high quality with good performance.
Our system turns out to be faster than the widely used hMetis system (see Section~\ref{Experiments}).

\paragraph{Outline and Contribution.} 
After giving a brief overview of related work in \autoref{RelatedWork} and introducing basic notation in \autoref{Preliminaries}, we explain how to compute a $k$-way partition via recursive bisection in \autoref{kwayRb}. 
We show how to derive good values for the balance constraint of each bipartitioning subproblem so that in the end, the balance constraint for the $k$-way partition is fulfilled. 
\autoref{Coarsening} describes our coarsening strategy. 
A carefully chosen rating function evaluates how attractive it is to 
contract two vertices. We present an effective strategy to limit the cost for reevaluating the rating function. 

The initial partitioner described in \autoref{InitialPartitioning} is based on a large portfolio of simple algorithms, each with some randomization aspect (fully random, BFS, label propagation, and nine variants of greedy hypergraph growing). 
Since these partitioners are very fast and only applied to a small core problem, we can afford to make a large number of attempts taking the best partition as the basis for further processing. 

Local improvement steps are expensive since many steps are needed and a naive implementation of the established techniques needs work proportional to the \emph{squares} of the net sizes. We integrate several techniques for reducing this bad behavior for large nets and additionally develop a way to \emph{cache} gain values to further reduce search overhead in \autoref{localizedFM}.
In \autoref{hypergraphDS} we introduce a hypergraph data structure that supports fast contraction and uncontraction of vertex pairs.
We evaluate our algorithm and the main competitors on a broad range of hypergraphs derived from well established benchmark sets. 
%such as the ISPD98 Circuit Benchmark Suite~\cite{ISPD98}, the University of Florida Sparse Matrix Collection~\cite{FloridaSPM} as well as the international SAT Competition 2014~\cite{SAT14Competition}.
The experiments reported in \autoref{Experiments} indicate that our algorithm computes partitions that are significantly smaller than hMetis and PaToH with considerably larger improvements for special instance classes like social networks, for bipartitioning, and partitions with an allowed imbalance of 10\%. At the same time, our algorithm is faster than hMetis.
\autoref{Conclusions} concludes the paper.

\section{Related Work} \label{RelatedWork}
Since the 1990s HGP has evolved into a broad research area. We therefore refer to  \cite{Papa2007,trifunovic2006parallel,Alpert19951,DBLP:conf/dimacs/2012} for an extensive overview. Here, we focus on issues closely related to the contributions of our paper.
The two most widely used general-purpose tools are PaToH~\cite{PaToH} (originating from scientific computing) 
and hMetis~\cite{hMetisRB,hMetisKway} (originating from VLSI design). 
Other software packages with certain distinguishing characteristics are known, in particular Mondriaan~\cite{Mondriaan} (sparse matrix 
partitioning), MLPart~\cite{MLPart} (circuit partitioning), Zoltan~\cite{Zoltan} and Parkway~\cite{Parkway2.0} (parallel), 
and UMPa~\cite{DBLP:conf/dimacs/CatalyurekDKU12} (directed hypergraph model, multi-objective).
All of these tools are based on the multilevel paradigm and compute a $k$-way partition either directly~\cite{hMetisKway,Parkway2.0,DBLP:conf/dimacs/CatalyurekDKU12,Aykanat:2008} or via recursive bisection~\cite{PaToH,hMetisRB,MLPart,Mondriaan,Zoltan}. 
The two most popular local search approaches are greedy algorithms~\cite{hMetisKway,DBLP:conf/dimacs/CatalyurekDKU12} or variations of
the Fiduccia-Mattheyses (FM) heuristic~\cite{FM82}. FM-type algorithms move vertices to other 
blocks in the order of improvements in the objective. Unlike simple greedy methods, FM can escape
local optima to some extent since it allows to worsen the objective temporarily. Partitioners based on recursive bisection use FM-based local search algorithms~\cite{PaToH,hMetisRB,MLPart,Mondriaan,Zoltan}, while direct $k$-way hypergraph partitioners employ greedy methods~\cite{hMetisKway,Parkway2.0,DBLP:conf/dimacs/CatalyurekDKU12,Aykanat:2008}.

$n$-level algorithms have been used in geometric data structures based on randomized incremental construction~\cite{nLevelGeometryNearestNeighbour,nLevelGeometryDelaunay} and as a preprocessing technique for route planning~\cite{nLevelCH}. More specifically, KaHyPar owes many basic ideas to its graph partitioning (GP) ancestor KaSPar~\cite{nGP}. However,
the implementation, the required design choices, and even the overall outcome are very different. KaSPar is a direct $k$-way partitioner that suffers from data structure overheads and the difficulty to integrate advanced global improvement methods based on flows \cite{SanSch11} and shortest paths \cite{SanSch13}. In contrast, KaHyPar is based on recursive bipartitioning and currently seems to be the method of choice for a wide range of hypergraph partitioning tasks. In particular, it is actually among the fastest codes available -- it seems that the overheads of a dynamic graph data structure are outweighed by many other complex issues in hypergraph partitioning that are unaffected by an $n$-level approach. There is a previous attempt on hypergraph bipartitioning in a Bachelor thesis \cite{ZieglerNHGP} with the approach to contract one hyperedge in each level.
However, that system was very slow so that we decided to start from scratch. Our second attempt was a direct $k$-way $n$-level partitioner
\cite{DirectKwayNhgp}. Despite several interesting ideas and best quality in the majority of experiments,
the $k$-way algorithm has not been able to improve on the state of the art consistently in terms of the time-quality trade-off.
However, we learned from that paper that recursive bipartitioning seems to be advantageous and thus decided to first focus on the highly optimized $n$-level recursive bipartitioner presented here. 

\section{Preliminaries} \label{Preliminaries}
An \textit{undirected hypergraph} $H=(V,E,c,\omega)$ is defined as a set of $n$ vertices $V$ and a
set of $m$ hyperedges $E$ with vertex weights $c:V \rightarrow \mathbb{R}_{\geq0}$ and hyperedge 
weights $\omega:E \rightarrow \mathbb{R}_{>0}$, where each hyperedge is a subset of the vertex set $V$ (i.e., $e \subseteq V$). In the
HGP literature, hyperedges are also called \emph{nets} and the vertices of a net are called \emph{pins} \cite{PaToH}.
We extend $c$ and $\omega$ to sets, i.e., $c(U) :=\sum_{v\in U} c(v)$ and $\omega(F) :=\sum_{e \in F} \omega(e)$.
A vertex $v$ is \textit{incident} to a net $e$ if $ \{v\} \subseteq e$. $\mathrm{I}(v)$ denotes the set of all incident nets of $v$. 
The \textit{degree} of a vertex $v$ is $d(v) := |I(v)|$.
Two vertices are \textit{adjacent} if there exists a net $e$ that contains both vertices. The set $\Gamma(v) := \{ u~|~\exists~e \in E : \{v,u\} \subseteq e\}$ denotes the neighbors of $v$.
The \textit{size} $|e|$ of a net $e$ is the number of its pins. Nets of size one are called \emph{single-node} nets. If $e_i = e_j$ we call nets $e_i$ and $e_j$ \emph{parallel}.

A \emph{$k$-way partition} of a hypergraph $H$ is a partition of its vertex set into $k$ \emph{blocks} $\mathrm{\Pi} = \{V_1, \dots, V_k\}$ 
such that $\bigcup_{i=1}^k V_i = V$, $V_i \neq \emptyset $ for $1 \leq i \leq k$ and $V_i \cap V_j = \emptyset$ for $i \neq j$. We use $b[v]$ to refer to the block id of vertex $v$.
A $k$-way partition decomposes the hypergraph into $k$ \emph{section hypergraphs}~\cite{Berge:75} $H \times V_i :=(V_i, \{e \in E~|~e \subseteq V_i\})$.
We call a $k$-way partition $\mathrm{\Pi}$ \emph{$\mathrm{\varepsilon}$-balanced} if each block $V_i \in \mathrm{\Pi}$ satisfies the \emph{balance constraint}:
$c(V_i) \leq L_{max} := (1+\varepsilon)\lceil \frac{c(V)}{k} \rceil$ for some parameter $\mathrm{\varepsilon}$. 
We call a block $V_i$ \emph{overloaded} if $c(V_i) > L_{max}$ and \emph{underloaded} if $c(V_i) < L_{max}$. Given a $k$-way partition $\mathrm{\Pi}$, the number of pins of a net $e$ in block $V_i$ is defined as 
$\mathrm{\Phi}(e,V_i) := |\{v \in V_i~|~v \in e \}|$. Net $e$ \emph{is connected} to block $V_i$ if $\mathrm{\Phi}(e, V_i) > 0$. Similarly, a block $V_i$ is \emph{adjacent} to a vertex $v \notin V_i$ if $\exists~e \in  I(v): \mathrm{\Phi}(e, V_i) > 0$. $\mathrm{R}(v)$ denotes the set of all blocks adjacent to $v$.
We call a net \emph{internal} if $\mathrm{ \mathrm{\Phi}(e, i) = |e|}$ for one block $V_i$ and \emph{cut} net otherwise. Analogously, a vertex contained in at least one cut net is called \emph{border vertex}. 
The \emph{$k$-way hypergraph partitioning problem} is to find an $\varepsilon$-balanced $k$-way partition of a hypergraph $H$ that minimizes the \emph{total cut} $ \omega(E')$ for some $\varepsilon$,  where $E'$ is
the set of all cut nets. This problem is known to be NP-hard \cite{Lengauer:1990}.

\emph{Contracting} a pair of vertices $(u, v)$ means merging $v$ into $u$. We refer to $u$ as the \emph{representative} and $v$ as the \emph{contraction partner}. 
The weight of $u$ becomes $c(u) := c(u) + c(v)$.  We connect $u$ to the former neighbors $\Gamma(v)$ of $v$, by replacing 
$v$ with $u$ in all nets $e \in I(v) \setminus I(u)$. Furthermore we remove $v$ from all nets $e \in I(u) \cap I(v)$.
%If a contraction leads to parallel nets we choose one net $e_i$ to remain in $H$, add the weight of the nets that are parallel to $e_i$ to $\omega(e_i)$
%and remove them from the hypergraph. Single-node nets created by a contraction are removed from the hypergraph, since such nets can never become part of the cut.
\emph{Uncontracting} a vertex $u$ reverses the contraction. The uncontracted vertex $v$ is put in the same 
block as $u$ and the weight of $u$ is set back to $c(u) := c(u) - c(v)$.

\section{$n$-Level Hypergraph Partitioning}
We now present our main contributions. A high-level overview of our $n$-level hypergraph partitioning framework is provided in Algorithm~\ref{alg:nHGP}. 
In \autoref{kwayRb}, we start by explaining how to compute a $k$-way partition via recursive bisection.
As other multilevel algorithms our algorithm has a coarsening, initial partitioning and an uncoarsening phase.
During the coarsening phase, we successively shrink the hypergraph by contracting only \emph{a single pair} of vertices \emph{at each level},
until it is small enough to be initially partitioned. We describe the details of our coarsening
algorithm in \autoref{Coarsening} and briefly discuss our portfolio-based initial partitioning approach in \autoref{InitialPartitioning}. The initial solution is transfered
to the next finer level by performing a \emph{single} uncontraction step. Afterwards, our localized local search algorithm described in
\autoref{localizedFM} is used to further improve the solution quality. All algorithms use the hypergraph data structure described in
\autoref{hypergraphDS}.

\subsection{$k$-way Partitioning via Recursive Bisection.} \label{kwayRb}
\begin{algorithm2e}[t!]
\caption{Algorithm Overview}\label{alg:nHGP}\normalsize
\LinesNumberedHidden
\SetKwFunction{partition}{partition}\SetKwFunction{proc}{proc}
\SetKwProg{myalg}{Algorithm}{}{}
\KwIn{Hypergraph $H$, lowest block id $k_{l}$, highest block id $k_{h}$, imbalance parameter $\varepsilon$.}

\myalg{\partition{$H:=(V,E), \varepsilon, k_{l}, k_{h}$}}{
  $k=k_{h}-k_{l}+1$\\
  \textit{// partition $H$ into $k$ blocks with block ids $k_l, ... ,k_h$}
  $ \mathrm{\Pi}_k := \emptyset$ \\
  \lIf() {$k_{l} = k_{h}$} {
    $\mathrm{\Pi}_k := V$; \Return $ \mathrm{\Pi}_k$
  }
 
  $\varepsilon' := $ \textit{ calculate according to Theorem 4.1}\\
  \textit{// coarsening phase} \\
  \While(){$H$ is not small enough} {
    \textit{// choose vertex pair with highest rating} \\
    $(u,v) := \argmax_{u \in V} score(u)$
    $H := \FuncSty{contract}(H,u,v)$ \Remi{$H := H \setminus \{v\}$}  
  }
  
  \textit{// initial partitioning phase}
  $\mathrm{\Pi}_2=(V_0,V_1) := \FuncSty{computeBisection}(H,\varepsilon')$ 
  
  \textit{// uncoarsening and local search phase}\\
  \While()
  {$H$ is not completely uncoarsened} {
    $(H, \mathrm{\Pi}_2, u,v) := \FuncSty{uncontract}(H,\mathrm{\Pi}_2)$
    
    $(H, \mathrm{\Pi}_2) := \FuncSty{localSearch}(H, \mathrm{\Pi}_2, u, v, \varepsilon')$
  }
  \textit{// recurse on section hypergraphs}
  $\mathrm{\Pi}_k := \mathrm{\Pi}_k~\cup~$\partition{$H\times V_0, k_{l}, k_{l} + \lfloor k/2 \rfloor -1, \varepsilon$}\\
  $\mathrm{\Pi}_k := \mathrm{\Pi}_k~\cup~$\partition{$H\times V_1, k_{l} + \lfloor k/2 \rfloor, k_{h}, \varepsilon$ }
  return $\mathrm{\Pi}_k$
}{}
\KwOut{$\varepsilon$-balanced $k$-way partition $\mathrm{\Pi}=\{V_1, \dots, V_k\}$}
\end{algorithm2e}
There are two approaches for computing a $k$-way partition within the multilevel framework. 
In \emph{direct} $k$-way partitioning, the initial partitioning algorithm computes a $k$-way partition, which is then
improved during uncoarsening using $k$-way local search algorithms. The most commonly used approach in HGP, however,
is to use recursive bisection~\cite{TR-02-025}. If $k$ is a power of two, the final $k$-way partition is
obtained by first computing a bisection of the initial hypergraph and then recursing on each of the two blocks.
Thus it takes $\log_2(k)$ such phases until the hypergraph is partitioned into $k$ blocks. 
If $k$ is not a power of two, the approach has to be adapted to produce appropriately sized partitions.
Our algorithm uses the following technique to compute a $k$-way partition via recursive bisection for arbitrary values of $k$:
We bisect the hypergraph such that one block has a maximum weight of $(1+\varepsilon') \lceil \lfloor k/2 \rfloor/k~c(V) \rceil$
and the other block has a maximum weight of $(1+\varepsilon') \lceil \lceil k/2 \rceil/k~c(V) \rceil$,
where $\varepsilon'$ is an adapted imbalance parameter that ensures that the final $k$-way partition is $\varepsilon$-balanced.
The former block is then partitioned recursively into $k' := \lfloor k/2 \rfloor$ blocks, while the latter is partitioned into $k' := \lceil k/2 \rceil$ blocks.
After each bisection step, we therefore have to solve two $k'$-way hypergraph partitioning problems.
The new imbalance parameter $\varepsilon'$ is chosen according to the following lemma:
% To do so, we 
%treat each of the two blocks as an independent hypergraph. 

\begin{AdaptiveE}
Let $H_0= H \times V_0 $ and $H_1= H \times V_1$ be the section hypergraphs induced by a bipartition $\mathrm{\Pi} = \{V_0,V_1\}$ of a hypergraph $H=(V,E,c,\omega)$ for
which we wish to compute an $\varepsilon$-balanced $k$-way partition. Using an adaptive imbalance parameter 
\[
\varepsilon' := \left( \left(1+\varepsilon \right) \frac{k' \cdot c(V)}{k \cdot c(V_i)}\right)^{\frac{1}{\lceil \log_2(k') \rceil}} -1
\]
to compute a $k'$-way partition (with $k'\geq 2$) of a hypergraph $H_i$ via recursive bisection ensures that the final $k$-way partition of $H$ is $\varepsilon$-balanced. When computing the very first bisection for a $k$-way partition, we set $H_0=H$, $k'=k$ and therefore  $\varepsilon' := (1+\varepsilon)^{(1/\lceil \log_2(k) \rceil)} -1$. 
\end{AdaptiveE}

\begin{proof}
(Outline) To show that using $\varepsilon'$ at each bisection step ensures an $\varepsilon$-balanced $k$-way partition,
 we use a maximum block weight $L_{max}' := (1+\varepsilon) \frac{c(V)}{k} \leq L_{max}$. If the weight of each of the $k'$ blocks of the $k'$-way partition is below $L_{max}'$, then
the final $k$-way partition of $H$ is $\varepsilon$-balanced. To ensure this, we have to determine the maximum possible weight one of these blocks can have.
Because $H_i$ is split at each bisection such that one block can be further divided into $\lfloor k'/2 \rfloor$ blocks
while the other is further split into $\lceil k'/2 \rceil$ blocks, the vertices of at least two blocks in the final $k'$-way partition have to be part of $\lceil \log_2(k') \rceil$ bisections.
Let $V_{max}$ be such a block and assume without loss of generality that at each bisection step, block $V_{max}$ has the maximum possible weight.
Using the initial imbalance parameter $\varepsilon$ at each bisection step would therefore result in a final block weight of
\begin{equation} \label{WCblockWeight}
c(V_{max}) := (1+\varepsilon)^{\lceil\log_2(k')\rceil}~\frac{c(V_i)}{k'} 
\end{equation}
In order to ensure that the original $k$-way partition of $H$ is $\varepsilon$-balanced, we therefore have to choose $\varepsilon$ in \autoref{WCblockWeight} such that
$c(V_{max}) \leq L_{max}'$.

Thus when recursively partitioning a section hypergraph $H_i$ with weight $c(V_i)$ into $k'$ blocks we choose a new imbalance parameter $\varepsilon'$ as follows:
\begin{equation}
\begin{split}
(1 + \varepsilon')^{\lceil \log_2(k') \rceil} \frac{c(V_i)}{k'} &\leq L_{max}' :=  (1+\varepsilon) \frac{c(V)}{k} \\
\Rightarrow \varepsilon' &\leq \left((1+\varepsilon)\frac{k'~c(V)}{k~c(V_i)}\right)^{\frac{1}{\lceil \log_2(k') \rceil}} -1
\end{split}
\end{equation}
%\qed
\end{proof}

Note that by definition, the section hypergraphs do not contain any cut nets. The cut nets of each bipartition can be discarded, 
since they will always be cut nets in the final $k$-way partition and already contribute $\omega(e)$ to the 
total cut size~\cite{PaToHManual}. This simultaneously reduces the number of nets as well as their average size in each section hypergraph,
without affecting the partitioning objective.

\subsection{Coarsening.} \label{Coarsening}
The goal of the coarsening phase is to contract highly connected vertices such that the number of nets remaining in the hypergraph and their size is
successively reduced \cite{hMetisRB}. Removing nets leads to simpler instances for initial partitioning, while small net sizes allow
 FM-based local search algorithms to identify moves that improve solution quality. 
Thus, we choose the vertex pairs $(u,v)$ to be contracted according to a rating function. 
More precisely, we adopt the \emph{heavy-edge} rating
 function also used by hMetis~\cite{hMetisRB}, Parkway~\cite{Parkway2.0} and PaToH~\cite{PaToHManual}, which prefers vertex 
pairs that have a large number of heavy nets with small size in common. 
However, in contrast to these tools, we additionally scale this score inversely with the product of the vertex weights $c(v)$ and $c(u)$ to keep the vertex weights of the coarse hypergraph reasonably uniform: 

\begin{equation}
r(u,v) := \frac{1}{{c(v) \cdot c(u)}}~\sum \limits_{e \in \{I(v) \cap I(u)\}}  \frac{\omega(e)}{|e| - 1}.
\end{equation}
This is also done in the matching-based coarsening algorithm of MLPart~\cite{MLPart} for similar reasons.

\paragraph{Algorithm Outline.}
At the beginning of the coarsening algorithm, all vertices are rated, i.e., for each vertex $u$ we 
compute the ratings of all neighbors $\mathrm{\Gamma}(u)$ and choose the vertex $v$ with the highest rating as contraction partner for $u$.
Ties are broken randomly to increase diversification. 
For each vertex, we then insert the vertex pair with the highest score into an addressable priority queue (PQ) using the rating score as key. 
This allows us to efficiently choose the best-rated vertex pair that should be contracted next. In each iteration, we remove the pair $(u,v)$ with the highest score and
contract it. After contraction, the entry of $v$ is removed from the PQ, since $v$ is no longer contained in the hypergraph. 

A contraction operation can lead to parallel nets (i.e., nets that contain exactly the same vertices) and single-node nets in $I(u)$.
In order to reduce the running time of the coarsening algorithm, we remove these nets from the hypergraph. Single-node nets are easily identified, because $|e|=1$. 
In case of parallel nets, we remove all but one from $H$. The weight of the remaining net $e$ is set to the sum of the weights of the nets that were parallel to $e$. Parallel nets are detected using an algorithm similar to the one in \cite{ParallelHEDetection}, which identifies vertices with identical structure in a graph: 
We create a fingerprint for each net $e \in \mathrm{I}(u)$: $f_i := \bigoplus_{v \in e} v\oplus x$, for some seed $x$.\footnote{$\oplus$ is the bitwise XOR} 
These fingerprints are then sorted and a final scan then identifies parallel nets: If two consecutive fingerprints $f_i, f_j$ are identical, we check whether the nets are truly parallel by comparing their pins\footnote{In principle, this approach could be accelerated using hashing. Parallel net detection, however, did not significantly contribute to the overall running time of our algorithm.}. 
Each contraction potentially influences the ratings of some neighbors $\mathrm{\Gamma}(u)$. The \emph{full} re-rating strategy therefore recalculates these ratings 
and updates the PQ accordingly. To avoid imbalanced inputs for the initial partitioning phase, vertices $v$ 
with $c(v) > c_{max} :=  s  \cdot \lceil \frac{c(V)}{t} \rceil$ are never allowed to participate in a contraction step and are thus removed from the PQ.
The coarsening process is stopped as soon as the number of vertices drops below $t$ or no eligible vertex is left. 
Parameter $s$ is used to favor the contraction of highly connected vertices.
Both parameters will be chosen in \autoref{Methodology}.

\paragraph{Advanced Update Strategy.}
Continuously re-rating the neighbors $\mathrm{\Gamma}(u)$ of a representative $u$ is the most expensive part of the algorithm: After each contraction, 
we have to look at all pins of all incident nets $I(u)$. The re-rating can therefore easily become the bottleneck -- especially if $H$ contains large nets or high degree vertices. 
To improve the running time in these cases, we developed a variation of the \emph{full} strategy: Our \emph{lazy} strategy does not re-rate any vertices immediately after contraction. 
Instead, all adjacent vertices $\mathrm{\Gamma}(u)$ are marked as \emph{invalid}. 
If the PQ returns an invalid vertex, we recalculate its rating and update the priority queue accordingly.

\subsection{Initial Partitioning.} \label{InitialPartitioning}
Coarsening is performed until the coarsest hypergraph is small enough to be partitioned by an initial partitioning algorithm.
We use a portfolio of several algorithms to compute an initial bipartition. Each algorithm is run twenty times. We select the partition
with the best cut and lowest imbalance to be projected back to the original hypergraph. In case all partitions are imbalanced, we choose the
partition with smallest imbalance. The portfolio-based approach increases diversification and produced the best results in \cite{HeuerInitialPartitioning}. In the following, we give a brief overview of
the algorithms used and refer to \cite{HeuerInitialPartitioning} for a more detailed description and evaluation. \emph{Random partitioning} randomly assigns
vertices to a block. \emph{Breadth-First-Search} (BFS) starts with a randomly chosen vertex and performs a BFS traversal
of the hypergraph until it has discovered half of the hypergraph. The vertices visited during the traversal constitute block $V_0$,
all remaining vertices constitute $V_1$. Furthermore, we use different variations of \emph{Greedy Hypergraph Growing} (GHG)~\cite{PaToH}.
Each version first computes two pseudo-peripheral vertices: Starting from a random vertex, we perform a BFS. The last vertex visited serves as the start vertex
for the next BFS. This vertex and the last vertex visited by the second BFS are supposed to be ''far'' away from each other. Therefore one is used as the seed vertex for block $V_0$, the other for block $V_1$. 
For each block, we maintain a PQ that stores the neighboring vertices of the growing cluster according to a score function.
We use the FM gain, which will be described in the next section, as well as the \emph{Max-Net} and \emph{Max-Pin} gain definitions
(which are also used in PaToH~\cite{PaToH}) as score functions.
Our GHG variants also differ in the way the clusters are grown.
\emph{Greedy-Global} always moves the vertex with the highest score in both PQs to the corresponding block.
\emph{Greedy-Sequential} first grows block $V_0$ and then block $V_1$, while \emph{Greedy-Round-Robin} grows both blocks simultaneously.
This again increases diversification. In total, we thus have nine different initial partitioning algorithms based on GHG. 
The last algorithm is based on \emph{size-constrained label propagation (SCLaP)} \cite{LPAgraphPartitioning}. Each vertex has a label representing its block. 
Initially all labels are empty. The algorithm starts by searching two pseudo peripheral vertices via BFS. One vertex along with $\tau$ of its neighbors gets label $V_0$, while the other 
vertex and $\tau$ of its neighbors get label $V_1$. 
We then perform label propagation until the algorithm has converged, i.e. no empty labels remain. Vertices with the same label then become the blocks of the partition.
The tuning parameter $\tau$ is used to prevent labels from completely disappearing in the course of the algorithm.
Based on the results in \cite{HeuerInitialPartitioning},  we set $\tau$ to five in our experiments.

\subsection{Localized FM Local Search.} \label{localizedFM}
Our local search algorithm is similar to the FM-algorithm~\cite{FM82} and is further inspired by the algorithm used in KaSPar~\cite{nGP} for graph partitioning.
A key difference to traditional FM is the way a local search pass is started: Instead of initializing the algorithm with all vertices or all border vertices, we perform a highly localized 
 search starting only with the representative and the just uncontracted vertex. The search then gradually expands around this vertex pair by successively considering neighboring vertices. 
Traditional multilevel FM implementations as well as KaSPar \emph{always} compute the gain of \emph{each} vertex \emph{from scratch} at \emph{each} level of the hierarchy.
During an FM pass, these values are then kept up-to-date by delta-gain-updates~\cite{PaToH}. Since our algorithm starts only around two vertices,
many gain values would never be used during a local search pass. We therefore maintain a \emph{gain cache} that ensures that the gain of a vertex move is calculated at most \emph{once} during \emph{all} local searches along the $n$-level hierarchy. 

\paragraph{Algorithm Outline.} 
We use two priority queues (PQs) to maintain the possible moves for all vertices -- one for each block. 
At the beginning of a local search pass, all queues are empty and disabled. A disabled PQ will not be considered when searching for the next move with the highest gain. 
All vertices are labeled inactive and unmarked. Only unmarked vertices are allowed to become active. 
To start the local search phase after each uncontraction, we activate the representative and the just uncontracted vertex if 
they are border vertices. Otherwise, no local search phase is started. 
\emph{Activating} a vertex $v$ currently assigned to block $b[v]$ means that we calculate the \emph{gain} $g_i(v)$ for moving $v$ to the other block $V_i \in R(v) \setminus \{b[v]\}$ and 
insert $v$ into the corresponding queue $P_i$ using $g_i(v)$ as key. The gain $g_i(v)$ is defined as: 
\begin{equation} \label{eq:gain}
\begin{split}
g_{i}(v) := &\sum \limits_{e \in I(v)} \{ \omega(e) : \mathrm{\Phi}(e, i) = |e| - 1\} \\ - &\sum \limits_{e \in I(v)} \{ \omega(e) : \mathrm{\Phi}(e, b[v]) = |e|\}.
\end{split}
\end{equation}
After insertion, PQs corresponding to \emph{underloaded} blocks become enabled. Since a move to an overloaded block will
never be feasible, a queue corresponding to an overloaded block is left disabled. The algorithm then 
repeatedly queries only the \emph{non-empty, enabled} queues to find the move with the highest gain $g_{i}(v)$, breaking ties arbitrarily.
Vertex $v$ is then moved to block $V_i$ and labeled inactive and marked. 
We then update all neighbors $\mathrm{\Gamma}(v)$ of $v$ as follows: All previously inactive neighbors are activated as described above. 
Neighbors that have become internal are labeled inactive and the corresponding moves are deleted from the PQs.
Finally, we perform \emph{delta-gain-updates} for all moves of the remaining  active border vertices in $\mathrm{\Gamma}(v)$:
If the move changed the gain contribution of a net $e \in I(v)$, we account for that change by
incrementing/decrementing the gains of the corresponding moves by $\omega(e)$ using the delta-gain-update algorithm of Papa and Markov~\cite{Papa2007}.

%%%%%%%%%%%%%%
% submission
%
%%%%%%%%%%%%%%
% TR Version
%%%%%%%%%%%%%%
\begin{figure*}[t!] 
\centering
\includegraphics[width=.9\textwidth]{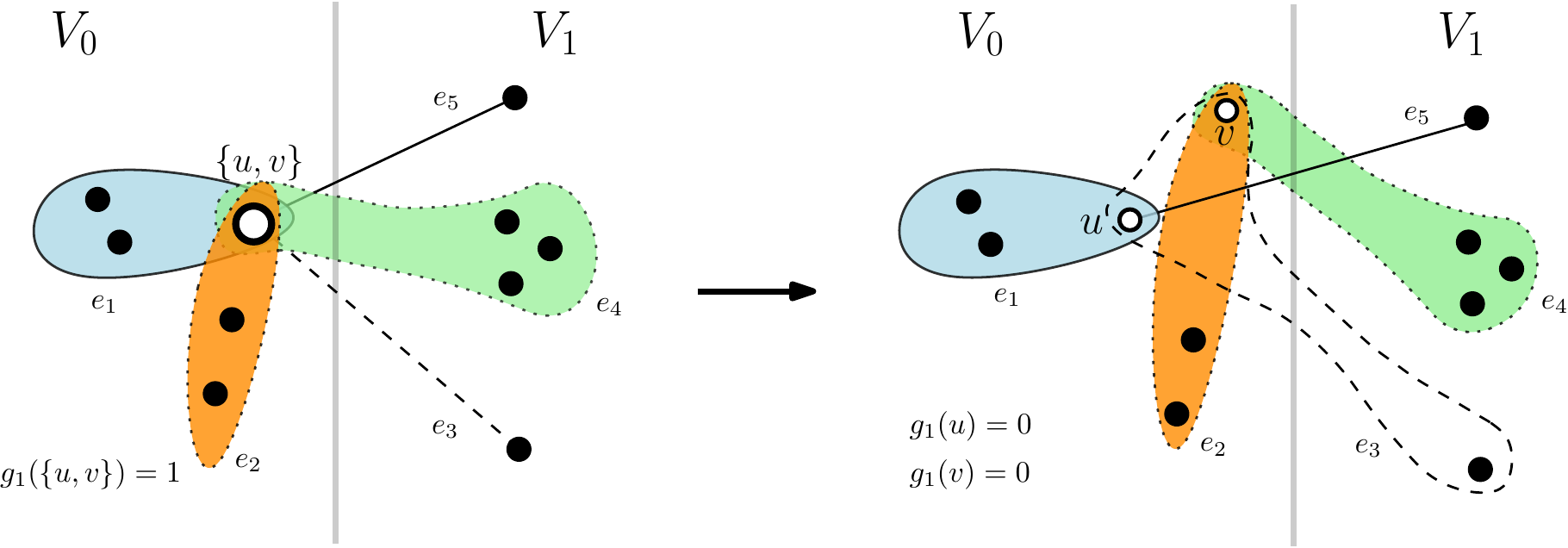}
\caption{Example of an uncontraction operation that affectes the cached gain value $c[u]$ of the representative $u$. Since nets $\{e_2,e_4\} \notin \mathrm{I}(u)$ 
after the uncontraction, they no longer contribute $-\omega(e_2)$ resp. $\omega(e_4)$ to the gain of $u$. Similarly, moving $u$ to $V_1$ now does not
remove $e_3$ from the cut anymore because of $v$. Its contribution to $c[u]$ therefore becomes zero.}\label{fig:gaincache}
\end{figure*}
%%%%%%%%%%%%%%

To further decrease the running time, we exclude nets from gain update that can not be removed from the cut in the current local search pass. 
A net is \emph{locked} in the bipartition, once it has at least one disabled pin in each of the two blocks~\cite{LockedNets}.
%%%%%%%%%%%%%%
% submission
% For more details, please refer to Appendix~\ref{app:lockednets}. 
%%%%%%%%%%%%%%
% TR Version
%%%%%%%%%%%%%%
In this case, it is not possible to remove such a net from the cut by moving any of the remaining movable pins to another block. 
Thus it is not necessary to perform any further delta-gain-updates for locked nets, since their contribution to the gain values of their pins does not change any more.
This observation was first described by Krishnamurthy~\cite{LockedNets}.
We integrate locking of nets into our algorithm by labeling each net during a local search pass. Initially, all nets are labeled \emph{free}.
Once the first pin of a net is moved, the net becomes \emph{loose}. It now has a pin in one block that cannot be moved again. Further moves to this 
block do not change the label of the net. As soon as another pin is moved to the other block, the net is labeled \emph{locked} and is
excluded from future delta-gain-updates. 
%%%%%%%%%%%%%%

Once all neighbors are updated, local search continues until either no non-empty, enabled PQ remains or a constant number of $c$ moves neither decreased the cut nor improved the current imbalance. 
The latter criterion is necessary, because otherwise the $n$-level approach could lead to $|V|^2$ local search steps in total. 
After local search is stopped, we reverse all moves until we arrive at the lowest cut state reached during the search that fulfills the balance constraint. 
All vertices become unmarked and inactive and the algorithm is then repeated until no further improvement is achieved.

\paragraph{Caching of Gain Values.}
We briefly outline the details of the gain cache. Let $c[v]$ denote the cache entry for vertex $v$.
After initial partitioning, the gain cache is empty. If a vertex becomes activated during a local search pass, we check whether or not its gain is already cached. 
If it is cached, the cached value is used for activation. Otherwise, we calculate the gain according to Eq.~(\ref{eq:gain}), %\autoref{eq:gain}
insert it into the cache and activate the vertex. After moving a vertex $v$ with gain $g_{i}(v)$ to block $V_i$, its cache value is set to $c[v] := - g_{i}(v)$. 
The delta-gain updates of its neighbors $\mathrm{\Gamma}(v)$ are then also applied to the corresponding cache entries. 
Thus the gain cache always resembles the current state of the hypergraph.
Since our algorithm performs a  rollback operation at the end of a local search pass that undoes vertex moves, we also have to undo delta-gain updates applied on the cache.
This can be done by additionally maintaining a \emph{rollback delta cache} that stores the negated delta-gain updates for each vertex. During rollback,
this delta cache is then used to restore the gain cache to a valid state.
Each time a local search is started with an uncontracted vertex pair $(u,v)$, we have to account for the fact that the uncontraction potentially affected~$c[u]$. 
A simple variant of the caching algorithm just invalidates the corresponding cache entry and re-calculates the gain. Since $v$ did not exist on previous levels of the hierarchy,
its gain must also be computed from scratch. 

%%%%%%%%%%%%%%
% submission
%In Appendix~\ref{app:moregaincache} we describe a more sophisticated variant that is able to update $c[u]$ based 
%on information gathered during the uncontraction and that further infers $c[v]$ from $c[u]$. 
%%%%%%%%%%%%%%
% TR Version
%%%%%%%%%%%%%%
We now describe a more sophisticated variant that is able to update $c[u]$ based on information gathered during the uncontraction and that further infers $c[v]$ from $c[u]$. 

After uncontraction, we initially set $c[v] := c[u]$. Both cache entries are then updated by examining each net $e \in \mathrm{I}(u)$. 
We have to distinguish three cases (see \autoref{fig:gaincache} for an example): 
\begin{enumerate}
\item After uncontraction, $u$ is not incident to net $e$ any more. If $e$ was a cut net that could have been removed from the cut by moving $u$ to the other block,
 $c[u]$ has to be decremented by $\omega(e)$ ($e_4$ in \autoref{fig:gaincache}). 
Similarly, if $e$ was an internal net,  moving $u$ would have made it a cut net. In this case, $c[u]$ is incremented by $\omega(e)$  
($e_2$ in \autoref{fig:gaincache}).
\item After uncontraction, $e$ contains both $u$ and $v$. If $\mathrm{\Phi}(e,b[v]) =2$, the net cannot be removed 
from the cut anymore by moving either $u$ or $v$. We therefore have to decrement both $c[u]$ and $c[v]$ by $\omega(e)$ ($e_3$ in \autoref{fig:gaincache}).
\item Finally, we have to account for nets to which $v$ is not incident ($e_1$ and $e_5$ in \autoref{fig:gaincache}). If such a net $e$ can be removed from the
cut by moving $u$, it contributes $\omega(e)$ to $c[u]$. We therefore have to decrement $c[v]$ by $\omega(e)$ to account for the fact that $e \notin I(v)$.
Similarly, if moving $u$ makes $e$ a cut net, we have to increment $c[v]$ accordingly.
\end{enumerate}
%%%%%%%%%%%%%%

\begin{figure*}[t] 
  \centering
  \includegraphics[width=.9\textwidth]{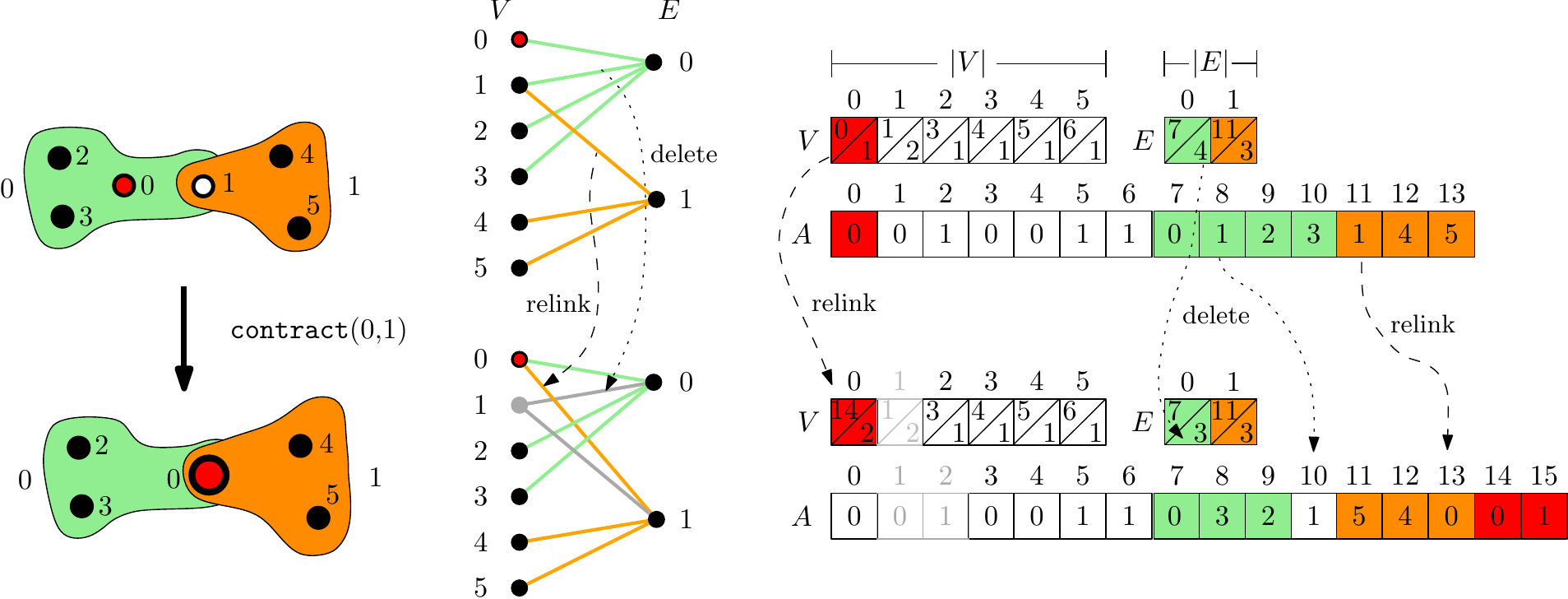}
  \caption{Example of a contraction operation. The hypergraph $H$ is depicted on the left, the corresponding bipartite graph representation is shown 
    in the middle and the adjacency data structure is shown on the right. The contraction leads to an edge deletion operation for net $0$ and
    a relink operation for net $1$.}\label{fig:graphds}
\end{figure*}

\subsection{Hypergraph Data Structure.} \label{hypergraphDS}
Conceptually, we represent the hypergraph $H$ as an undirected \emph{bipartite} graph $G=(V \dot \cup E, F)$. The vertices and nets
of $H$ form the vertex set. For each net $e$ incident to a vertex $v$, we add an edge $(e,v)$ to the graph.
The edge set $F$ is thus defined as $F := \{(e,v)~|~e \in E \wedge v \in e \}$.
In the following, we use \emph{nodes} and \emph{edges} when referring to $G$ and \emph{vertices} 
and \emph{nets} when referring to $H$. When contracting a vertex pair $(u,v)$, we mark the corresponding node $v$ as deleted. The edges $(v,e)$ incident to $v$ are treated as follows:
If $G$ already contains an edge $(u,e)$, then net $e$ contained both $u$ and $v$ before the contraction. 
In this case, we simply delete the edge $(v,e)$ from $G$. Otherwise, net $e$ only contained $v$. We therefore have to
relink the edge $(v,e)$ to $u$.
A modified adjacency array is used to represent $G$. It is divided into two offset arrays $V$, $E$ and an incidence array $A$ that stores 
the edges leaving $v$ for each node $v \in V \dot \cup E$. $V$ stores the starting positions of the entries in $A$ ($V[\cdot].f$) and $d(v)$ ($V[\cdot].s$). $E$ stores
the same for nodes representing the nets of $H$. Nets incident to a vertex $v$ are accessible as $A[V[v].f], ...,  A[V[v].f+V[v].s-1]$. The
pins of a net $e$ are accessed similarly using offset array $E$. An example is shown in \autoref{fig:graphds}.
Although the concept of relinking and removing edges remains the same for GP~\cite{nGP} and HGP, the actual implementation is different: 
While an adjacency array used as a graph data structure allows to access all neighbors $\Gamma(v)$ of node $v$, we can only access the nets $\mathrm{I}(v)$ of a vertex $v$ in HGP.
Since \cite{nGP} does not provide any implementation details, we give a brief description on how contraction and uncontraction operations are implemented in our hypergraph data structure.

%Deleting an edge $(v,e) \in G$ is realized by swapping $v$ with
%the last pin of net $e$ located at $A[E[e].f +E[e].s-1]$ and decrementing $E[e].s$. Relinking an edge $(v,e) \in G$ to $u$ takes three steps:
%(i) swapping $v$ with the last pin of $e$, (ii) adding $u$ to the pins of $e$ and simultaneously removing $v$ by setting $A[E[e].f + E[e].s -1 ] := u$,
%(iii) copying the subarray of $u$ to the end of $A$, appending $e$ and updating $V[u].f$ and $V[u].s$ accordingly.\footnote{This could
%  be improved by leaving some free space in $A$ for each node. The copying operations, however, had no noticeable effect on the running time of our algorithm.}
%%%%%%%%%%%%%%
% submission
%After a deletion or relink operation, we $v$ disable to remove it from the graph.
%\autoref{fig:graphds} shows an example. A more detailed description is provided in Appendix~\ref{app:ds}.
%%%%%%%%%%%%%%
% TR Version
%%%%%%%%%%%%%%
\begin{figure*}[t!]
\begin{minipage}[b!]{.48\textwidth}
  \vspace{0pt}
\begingroup
\removelatexerror  
\begin{algorithm2e}[H]
\caption{Contract}\label{alg:contraction}\normalsize
\KwIn{Vertex pair to be contracted $(u,v)$}
%\emph{// remember vertex pair and old state of$u$} \\
$\mathcal{M} := \{ u, v,  V[u].f, V[u].s\}$ \\		
$c(u) := c(u) + c(v)$\\
\emph{copy} $:=$ \textbf{True}

\ForEach() {$e \in I(v) $}{
  % \textit{//last pin slot of $e$}\\
  $\tau, l := E[e].f + E[e].s - 1$\\
  % \textit{//iterate over pins}\\
  \For() {$i:=E[e].f; i \leq l; \Inc i$} {
      % \textit{move v to the last pin slot and search for $u$}\\
      \If() {$A[i] = v$} { 
        $\FuncSty{swap}(A[i],A[l])$, $\Dec i$
 		} \lElseIf() {$A[i] = u$} {
 			 $\tau :=i$
 		}
 	}

 	\If(\emph{ // relink operation}) {$\tau = l$} {
          % \emph{//Net $e$ does not contain $u$}\\
          % \emph{//reuse slot of $v$ in net $e$ for $u$}\\
          $A[l] := u$\\	
          \If() {copy} {
            $A.\FuncSty{append}(\mathrm{I}(u))$\\
            $V[u].f := |A| - V[u].s$ \\
            \emph{copy} $:=$ \textbf{False}\\
          }
          $A.\FuncSty{append}(e)$
          
          $\Inc V[u].s$
 	} \Else(\emph{ // delete operation}) {
%          \emph{// net $e$ contains both $u$ and $v$: cut off $v$}\\
 	$\Dec E[e].s$
 	}

}
$V[v]$.\FuncSty{disable()}\\
\KwOut{Contraction memento $\mathcal{M}$}
\end{algorithm2e}
\endgroup
\end{minipage}%
\begin{minipage}[b!]{.48\textwidth}
  \vspace{0pt}
\begingroup
\removelatexerror
\begin{algorithm2e}[H]
\caption{Uncontract}\label{alg:uncontraction}\normalsize
\KwIn{Contraction Memento $\mathcal{M}$}
$V[\mathcal{M}.v]$.\FuncSty{enable()} \\
\emph{// bitset}\\
$b = [b_0,\dots, b_{m-1}] := [false, \dots, false]$  \\
\emph{// assume no net contained $u$}\\
\ForEach() {$e \in I(\mathcal{M}.v)$} {
$b[e] :=$ \textbf{True}
}

\emph{// these nets actually contained $u$}\\
\For() {$i:=\mathcal{M}.f; i < \mathcal{M}.f + \mathcal{M}.s; \Inc i$} {
$b[A[i]] :=$ \textbf{False}
}

\If() {$V[\mathcal{M}.u].s - \mathcal{M}.s > 0$} {
   \emph{// reverse relink operations}\\
  \ForEach() {$e \in I(\mathcal{M}.u)$} {
    \If() {$b[e]$} {
      % \emph{//net $e$ was not incident to $u$}\\
      \ForEach() {$ p \in e$} {
        % \emph{net $e$ was not incident to $u$}\\
        % \Remi{and undo relink by resetting it to $v$
        \lIf() {$p = \mathcal{M}.u$} {$ p:= \mathcal{M}.v$; break
        }
      }
    }
  }
}

$V[\mathcal{M}.u].f := \mathcal{M}.f$\\ % \Remi{restore representative $u$}\\
$V[\mathcal{M}.u].s := \mathcal{M}.s$ \\
$c(\mathcal{M}.u) := c(\mathcal{M}.u) - c(\mathcal{M}.v)$ \\

\ForEach(  \emph{// reverse deletes}) {$e \in I(\mathcal{M}.v)$} {
	\lIf() {$b[e] =$ \textbf{False}} {
	$\Inc E[e].size$
	}
}
\vspace{7.1pt}
\end{algorithm2e}
\endgroup
\end{minipage}%
\end{figure*}
\paragraph{Contraction.}
Contracting a vertex pair $(u,v) \in H$ works as follows: For each net $e \in \mathrm{I}(v)$ we have to determine if the corresponding 
edge $(v,e) \in G$ can simply be deleted or if a relink operation is necessary. This can be done with one iteration over the pins of $e$. During
this iteration, we swap $v$ with the last pin of $e$ located at position $A[E[e].f +E[e].s-1]$ and additionally search for vertex $u$. If we found $u$, then there is no
need to perform a relink operation and we can remove $v$ from $e$ by simply decrementing $E[e].s$. 
If $u$ was not found, we have to relink $e$ to $u$. A relink operation adds the undirected edge $(u,e)$ to $G$. 
To achieve this in our data structure, we have to add $e$ to the subarray of $u$ and vice versa. The latter
can be accomplished by reusing the pin slot of $v$: After the iteration over the pins of $e$, $v$ is the last entry in the subarray of $e$.
Setting $A[E[e].f + E[e].s -1 ] := u$ therefore adds $u$ to the pins of $e$ and simultaneously removes $v$. 
The former is more difficult, because we actually have to extend the subarray of $u$.
This is done by copying the subarray of $u$ to the end of $A$, appending $e$ and updating $V[u].f$ and $V[u].s$ accordingly.\footnote{This could
be improved by leaving some free space in $A$ for each node. The copying operations, however, had no noticeable effect on the running time of our algorithm.}
To complete the contraction operation, we disable $v$ to remove it from the hypergraph.
Algorithm~\ref{alg:contraction} gives the corresponding pseudocode. 
Note that the subarray of $u$ is copied to the end of $A$ at most once during a contraction as part of the first relink operation.
%%%%%%%%%%%%%%

\paragraph{Uncontraction.} To reverse a contraction,  we store a memento $\mathcal{M}$ for each contraction consisting of the contracted vertex pair $(u,v)$
and the values $V[u].f$ and $V[u].s$ before the contraction. After re-enabling $v$, deletions can be reversed by simply increasing $E[e].size$ for the corresponding nets.
To reverse a relink operation of an edge $(u,e)$, we first reset the pin slot containing $u$ to $v$. 
Setting $V[u].f := \mathcal{M}.f$ and $V[u].s := \mathcal{M}.s$ then restores $\mathrm{I}(u)$ to the state before the contraction.
Deletion operations have to be reversed for all nets $\bar{D}:=\mathrm{I(v)} \cap \{A[\mathcal{M}.f], ..., A[\mathcal{M}.f + \mathrm{M}.s-1]\}$, because these nets contained both $u$ and $v$. 
All remaining nets $\bar{R}:= \mathrm{I}(v) \setminus \bar{D}$ require the reversal of a relink operation. 
A pseudocode description of the uncontraction operation can be found in Algorithm~\ref{alg:uncontraction}.
After initial partitioning, we initialize $\mathrm{\Phi}(e,V_i)$ for each cut net $e$. For constant-time border vertex checks, we additionally store the number of incident cut nets for each vertex.
Both data structures are then maintained and updated during local search.

\section{Experiments} \label{Experiments}
\paragraph{System.}
All experiments are performed on a single core of a machine consisting of two Intel Xeon E5-2670 Octa-Core processors (Sandy Bridge) 
clocked at $2.6$ GHz. The machine has $64$~GB main memory, $20$ MB L3-Cache and 8x256 KB L2-Cache and is running Ret Hat Enterprise Linux (RHEL) 6.4.

\paragraph{Algorithm Configuration and Methodology.}  \label{Methodology}
The algorithm is implemented in the $n$-level hypergraph partitioning framework \emph{KaHyPar} (\textbf{Ka}rlsruhe \textbf{Hy}pergraph \textbf{Par}titioning). The code is written in C++ and compiled using gcc-4.9.1 with flags \texttt{-O3} \texttt{-mtune=native} \texttt{-march=native}.
We performed a large number of experiments to tune the parameters of our algorithms on medium-sized VLSI and sparse matrix instances. %and  $k \in \{2,4,8,16,32,64\}$.
The properties of these hypergraphs are summarized in Appendix~\ref{app:tuninginstances}. 
A full description of these experiments is omitted due to space constraints. The following decisions are made based on the results summarized in \autoref{tbl:sttuning} and \autoref{tbl:ctuning} in Appendix~\ref{app:parametertuning}: 
The coarsening process is stopped as soon as the number of vertices drops below $t=320$ or no eligible vertex is left. 
The scaling factor $s$ for the highest allowed vertex weight during coarsening is set to $3.25$. Each local search pass is stopped as soon as $c=350$ moves did not yield any improvement.

We perform ten repetitions with different seeds for each test instance and report the \emph{arithmetic 
  mean} of the computed cut and running time as well as the best cut found. When averaging over different instances, we use the 
\emph{geometric mean} in order to give every instance a comparable influence on the final result. In order to include instances with
a cut of zero into the results, we set the corresponding cut values to \emph{one} for geometric mean and ratio computations.
%The results of our experiments are reported in two different ways:
%To perform a detailed comparison, plots relate the smallest average cut of all partitioners to the cut of each partitioner individually on a 
%per-instance basis. For each algorithm, these ratios are sorted in increasing order. Note that these plots use a cube root scale for both axes to reduce right skewness~\cite{st0223}.\sschl{ok?}
%Running times are compared by plotting the ratio between the running time of each algorithm and the running time of the fastest algorithm on a 
%per-instance basis. Ratios are again sorted in increasing order. Note that these plots use a log-scaled y-axis.

\paragraph{Instances.} \label{Instances}
We evaluate our algorithm on hypergraphs derived from three %well established 
benchmark sets: The ISPD98 VLSI Circuit Benchmark Suite~\cite{ISPD98},
the University of Florida Sparse Matrix Collection~\cite{FloridaSPM} and the international SAT Competition 2014~\cite{SAT14Competition}.
From the latter, we randomly selected~$100$ instances from the application track and converted them into hypergraphs as follows:
Each boolean variable (and its complement) is mapped to one vertex and each clause constitutes a net~\cite{Papa2007}. 
The Sparse Matrix Collection is organized into $172$ groups and each group contains matrices of different application areas.
From each group, we choose one matrix for each application area that has between $\numprint{10000}$ and $\numprint{10000000}$ columns.
In case multiple matrices fulfill our criteria, we randomly select one. In total, we include $192$ matrices, 
which are translated into hypergraphs using the row-net model~\cite{PaToH}, i.e. each row is treated as a net and each column as a vertex.
Empty rows are discarded. Both vertices and nets have unit weight. 
Together with the $18$ VLSI instances derived from the ISPD98 Circuit Benchmark Suite, a total of $310$ hypergraphs constitute our benchmark set.
Each of these hypergraphs is partitioned into $k \in \{2,4,8,16,32,64,128\}$ blocks with $\varepsilon = 0.03$. 
For each value of $k$, a $k$-way partition is considered to be \emph{one} test instance, resulting in a total of $2170$ instances.
We exclude $173$ test instances, because either PaToH-Q could not allocate enough memory or other partitioners did not finish in time. 
The excluded instances are shown in Appendix~\ref{app:excludedinstnaces}. Note that only PaToH-D was able to partition all instances within the given time limit. 
The comparison in \autoref{sec:full_set_epsilon003} is therefore based on the remaining $1997$ test instances.
\begin{figure}[t!]
\centering
\begin{knitrout}
\definecolor{shadecolor}{rgb}{0.969, 0.969, 0.969}\color{fgcolor}
{\centering \includegraphics[width=.45\textwidth]{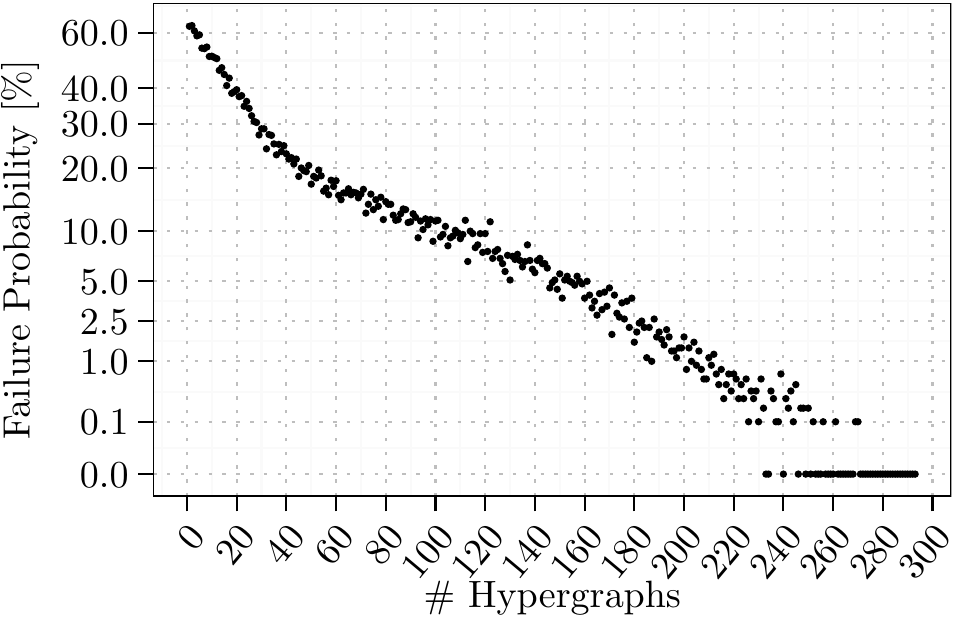} 
}
\end{knitrout}
\caption{Probability that a random sample of a certain number of hypergraphs does not provide the same results as the
full benchmark set. Each data point is based on \numprint{1000} random samples.}\label{fig:sampling}
\end{figure}

In order to evaluate the performance of our algorithm for different values of $\varepsilon$, we perform additional experiments
for a subset of the 293 hypergraphs used in \autoref{sec:full_set_epsilon003}, the same values of $k$, and $\varepsilon \in \{0.01, 0.1\}$. To estimate the 
number of hypergraphs necessary to produce the same qualitative results as presented in \autoref{sec:full_set_epsilon003}, we performed the following experiment:
For each subset size possible, we took a random sample of that size and verified whether or not the results for the sample match the results for the full benchmark set
by comparing the sorted min-cut ratios of KaHyPar to the ratios of all other partitioners. Only in case all ratios of KaHyPar were better than or equal to
the ratios of all other partitioners, we counted the sample as a successful trial. Note that this is a rather strong requirement for the selected subset. 
\autoref{fig:sampling} summarizes these experiments. Each data point corresponds to 1000 random samples. Based on these results, we use a subset of 100 hypergraphs
to be able to reproduce the results with a probability of 90\%. The hypergraphs are chosen as follows: 
We use the 10 largest VLSI hypergraphs (ibm09 - ibm18), 30 randomly chosen SAT hypergraphs and 60 randomly chosen sparse
matrix hypergraphs. This subset is used in the experiments presented in \autoref{sec:different_imbalance}.

%\paragraph{Partitioners used for Comparison.} \label{SystemsForComparison} 
We compare our algorithms to both the $k$-way (hMetis-K) and the recursive-bisection variant (hMetis-R) of hMetis 2.0 (p1)~\cite{hMetisRB,hMetisKway} and to PaToH 3.2~\cite{PaToH}. We choose these two tools because of the following reasons: 
% The authors Zoltan note that 
PaToH produces better quality than Zoltan's native parallel hypergraph partitioner (PHG) in serial mode~\cite{Zoltan,ZoltanUserGuide}.
Parkway does not run in serial mode and was found to be comparable to Zoltan PHG in serial mode~\cite{Zoltan}. 
Furthermore, PaToH has been shown to compute better solutions than Mondriaan~\cite{PaToHBetterThanMondriaansOwnPartitioner,PaToHBetterThanMondriaan} and MLPart~\cite{PaToHvsMLPart}.
Additionally, MLPart is restricted to bisections~\cite{MLPartBipartitioning,MLPartNoMultiway}. 
UMPa does not improve on PaToH when optimizing single objective functions that do not benefit from the directed hypergraph model~\cite{UMPa}.

hMetis-R defines the maximum allowed imbalance of a partition differently \cite{hMetisRB}: An imbalance value of 5, for example,
allows each block to weigh between $0.45 \cdot c(V)$ and $0.55 \cdot c(V)$ \emph{at each bisection step}.
We therefore translate our imbalance parameter $\varepsilon=0.03$ to $\varepsilon'$ as described in Eq. (\ref{eq:RBimbalance}) such that it matches our balance constraint after $\log_2(k)$ bisections: 
\begin{equation} \label{eq:RBimbalance}
  \varepsilon' := 100 \cdot \left(\left( (1+\varepsilon)~\frac{\lceil \frac{c(V)}{k} \rceil}{c(V)}\right)^{\frac{1}{\log_2(k)}} - 0.5 \right)
\end{equation}
PaToH is configured to use a final imbalance ratio of $\varepsilon \in \{0.01, 0.03, 0.1\}$ to match our balance constraint for the corresponding experiment.
Since it ignores the random seed if configured to use the quality preset, we report both the result of the quality preset (PaToH-Q) and the average over ten repetitions using the default configuration (PaToH-D). 
All partitioners have a time limit of $\numprint{15000} s$ per test instance.
The complete benchmark set, detailed statistics for each hypergraph and per-instance partitioning results for all experiments reported in this paper are publicly available~\cite{schlag_2015_29685}.

\paragraph{Effects of Engineering Efforts.}
Engineering the algorithms of the coarsening and local search phase is critical to the overall performance of our $n$-level hypergraph partitioner. 
We partitioned the hypergraph derived from the sparse matrix \texttt{wb-edu} into two blocks as an example. Using the \emph{lazy} re-rating strategy
for coarsening reduces the running time of the coarsening phase by \emph{two orders of magnitude}: While coarsening took $38349.6s$ using the full re-rating strategy,
the lazy re-rating strategy only took $450.3s$. In KaSPar~\cite{nGP} it was sufficient to calculate the gain values from scratch on each local search pass.
However, this approach leads to poor performance in $n$-level hypergraph partitioning, although our implementation already employs several known speedup techniques (delta-gain-updates, locked nets).
Without gain caching, local search took $2353.5s$. Activating the caching mechanism reduced the running time by a factor of \emph{four} down to $579.9s$.

\subsection{Comparison on full Benchmark Set.}\label{sec:full_set_epsilon003}
In $3700$ out of $19970$ experiments hMetis-K produced imbalanced partitions (up to $14\%$ imbalance). KaHyPar produced ten imbalanced partitions (up to $4\%$ imbalance) and one partition of PaToH-Q had $9\%$ imbalance. All instances partitioned by hMetis-R and PaToH-D fulfilled the balanced constraint
of $3\%$. In the following comparisons hMetis-K therefore has slight advantages because we do not disqualify imbalanced partitions.
\begin{figure*}
\centering
\begin{knitrout}
\definecolor{shadecolor}{rgb}{0.969, 0.969, 0.969}\color{fgcolor}

{\centering \includegraphics[width=.45\textwidth]{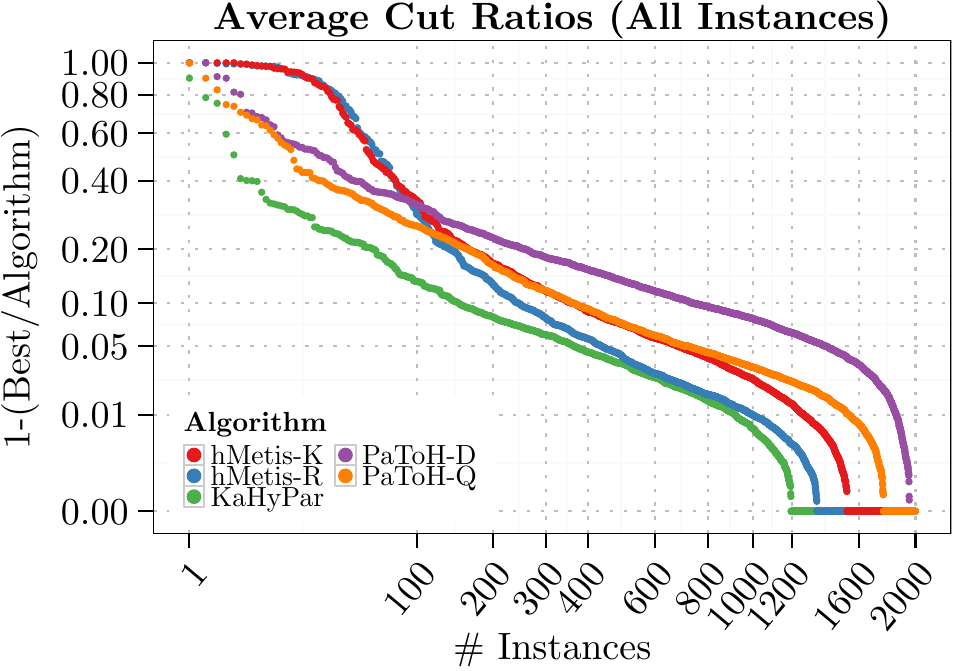} 
\includegraphics[width=.45\textwidth]{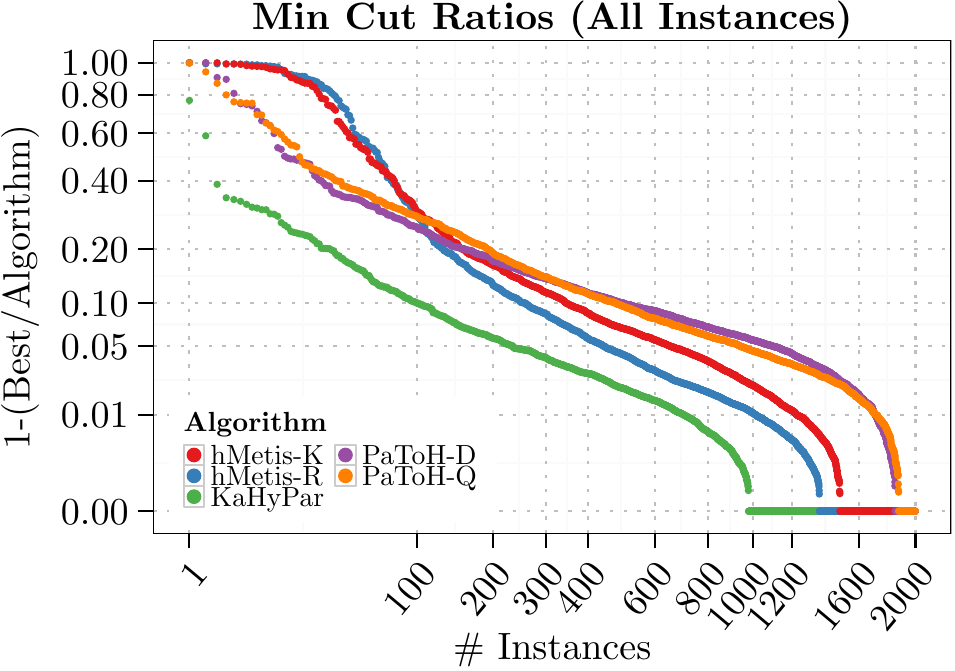} 
\includegraphics[width=.45\textwidth]{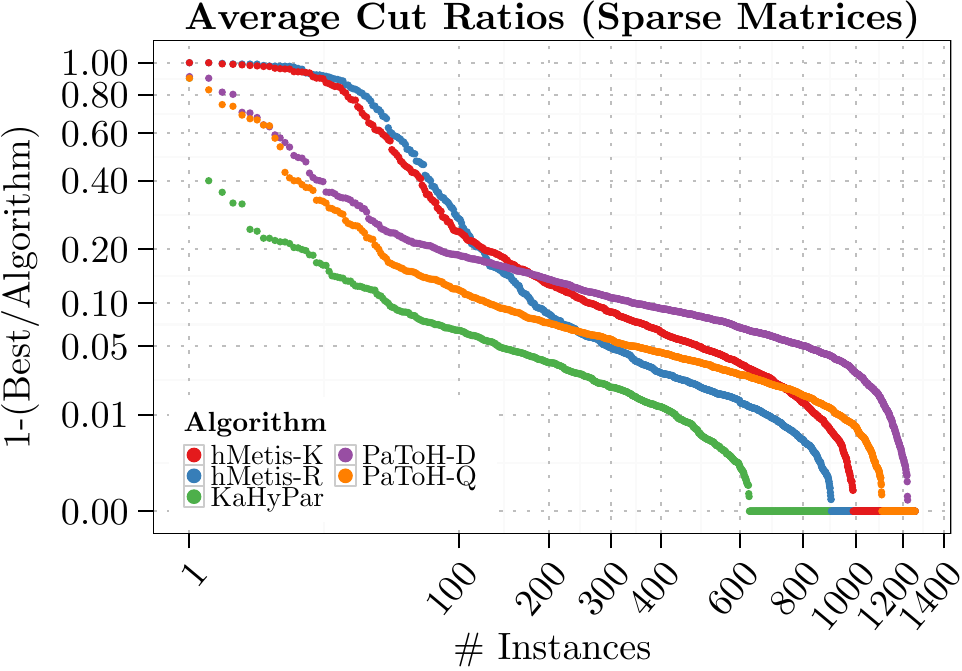} 
\includegraphics[width=.45\textwidth]{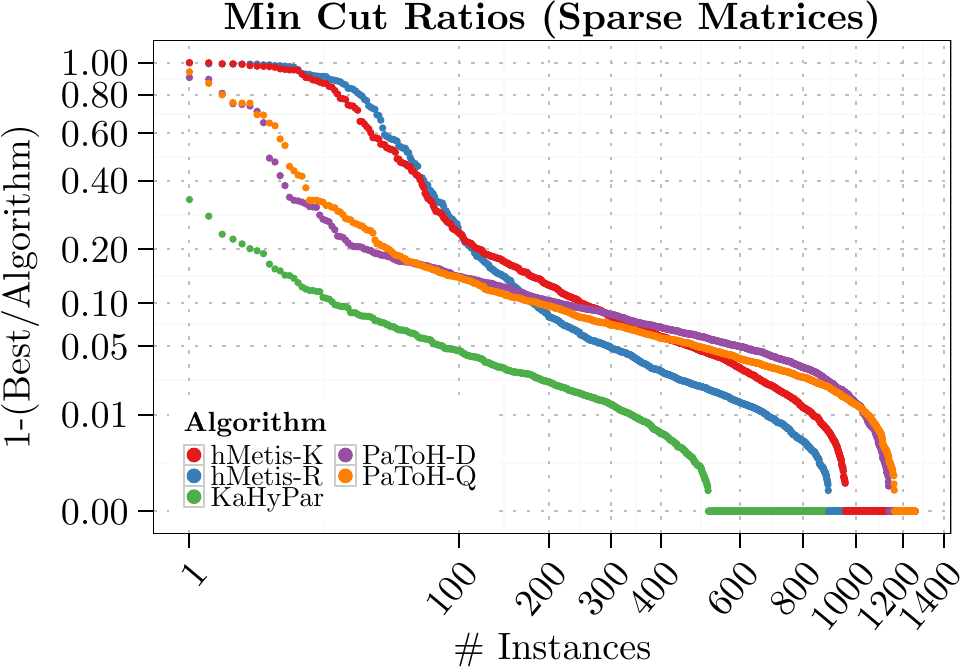} 
\includegraphics[width=.45\textwidth]{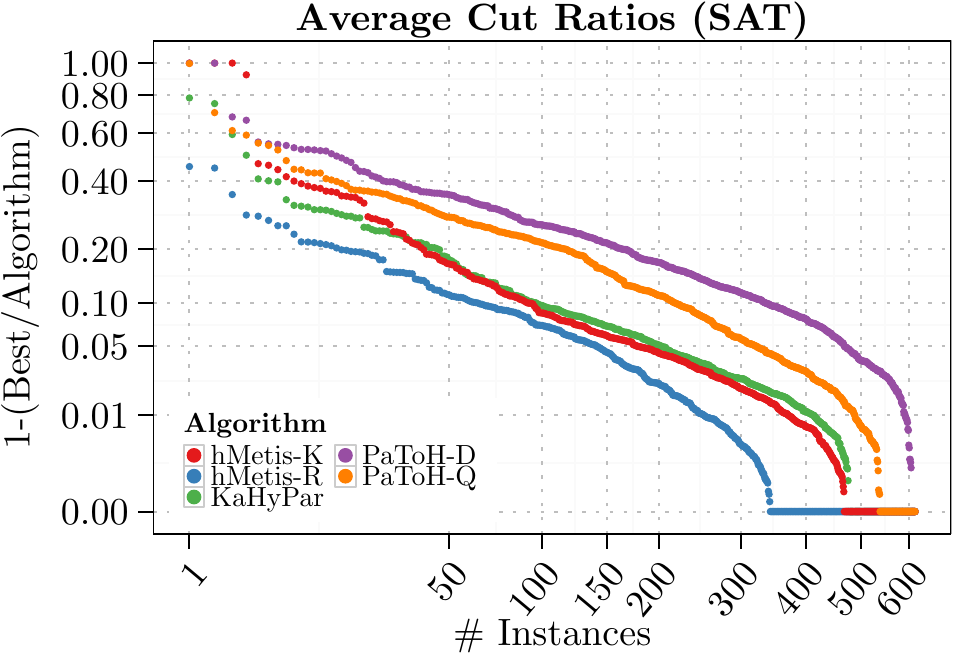} 
\includegraphics[width=.45\textwidth]{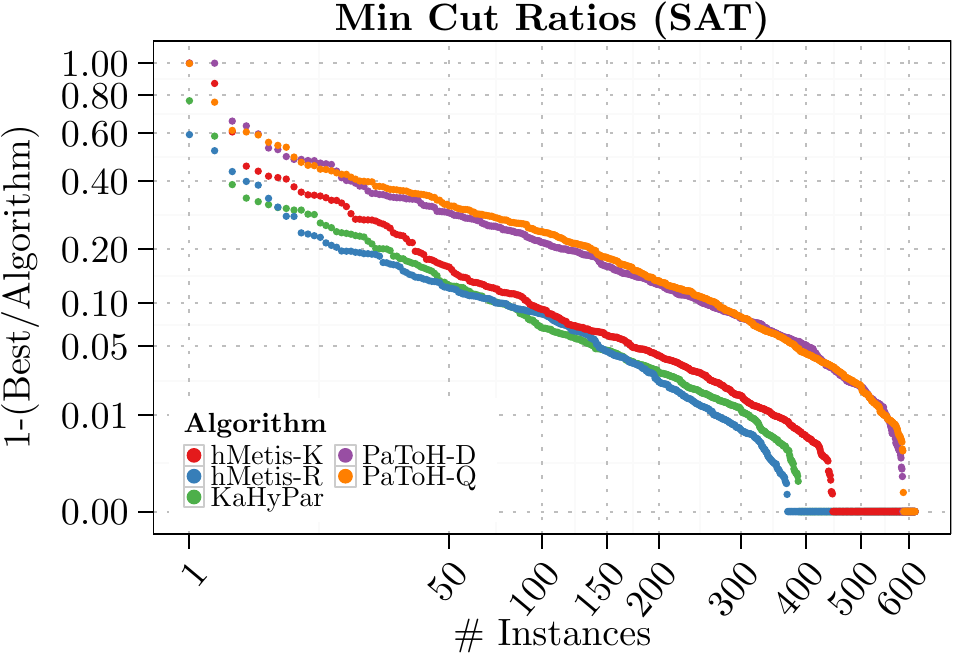} 
\includegraphics[width=.45\textwidth]{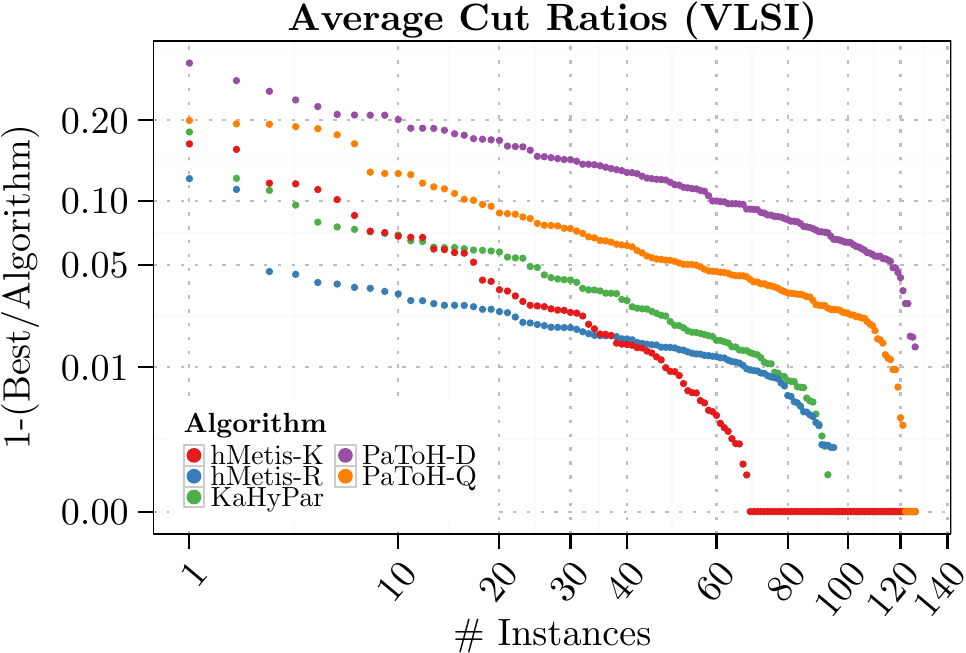} 
\includegraphics[width=.45\textwidth]{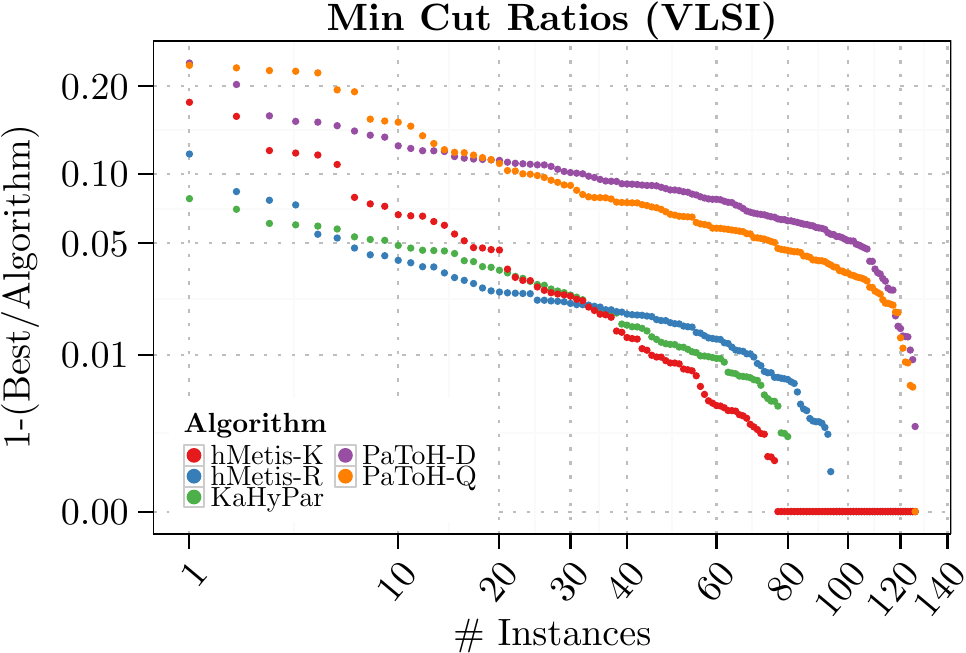} 

}
\end{knitrout}
\caption{Performance plots for $\varepsilon = 0.03$. The y-axis shows the ratio between the smallest cut of all algorithms and the cut produced
by the corresponding algorithm. Since we plot 1 - (best/algorithm), a value of zero indicates that the corresponding algorithm
produced the best solution. Note the cube root scale for both axes.}\label{fig:full_eval}
\end{figure*}

\autoref{fig:full_eval} summarizes the results of our experiments for $\varepsilon=0.03$. For each algorithm, it
 relates the smallest average (left column) and minimum (right column) cut of all algorithms to the corresponding cut
 produced the algorithm on a per-instance basis. For each algorithm, these ratios are sorted in increasing order. 
Note that these plots use a cube root scale for both axes to reduce right skewness~\cite{st0223}. Since the plot shows
$1-(\text{best}/\text{algorithm})$, a value of zero indicates that the corresponding algorithm produced the best solution.
A point close to one indicates that the partition produced by the corresponding algorithm was considerably worse than the
partition produced by the best algorithm. Thus an algorithm is considered to perform better than another algorithm, if its corresponding
ratio values are below those of the other algorithm.

Looking at the solution quality across all instances (top), KaHyPar outperforms all other systems
regarding both average and best solution quality. KaHyPar produced the best partitions for $1018$ of the $1997$ instances. It is followed by hMetis-R ($643$) and hMetis-K ($519$). 
PaToH-D computed the best partitions for $147$ and PaToH-Q for $123$ instances. Note that for some instances multiple partitioners computed the same solution. 
Comparing the best solutions of KaHyPar to each partitioner individually, KaHyPar produced better partitions than PaToH-Q, PaToH-D, hMetis-K, hMetis-R in $1742$, $1701$, $1247$ and $1146$ cases, respectively.
Note that hMetis-R outperforms hMetis-K, although we had to tighten the balance constraint for hMetis-R and allowed imbalanced solutions for hMetis-K.
Except for a small number of instances where of both hMetis variants compute partitions that are significantly inferior to the best solution, the hMetis variants perform better than
both PaToH variants. 
The running times of each partitioner are compared in \autoref{fig:running_time_overview}. The plots show the ratio between the running time of each algorithm 
and the running time of the fastest algorithm on a per-instance basis. Ratios are again sorted in increasing order. Note that these plots use a log-scaled y-axis.
PaToH is the fastest partitioner. Reasons for this are not only the smaller number of levels but also aggressive optimizations like ignoring large hyperedges during coarsening 
and a memory management that makes PaToH-Q fail for more than 3 \% of the original instances. 
Taking \emph{both} running time and quality into account KaHyPar dominates hMetis-R and hMetis-K.

\begin{figure*}
\centering
\begin{knitrout}
\definecolor{shadecolor}{rgb}{0.969, 0.969, 0.969}\color{fgcolor}

{\centering \includegraphics[width=.45\textwidth]{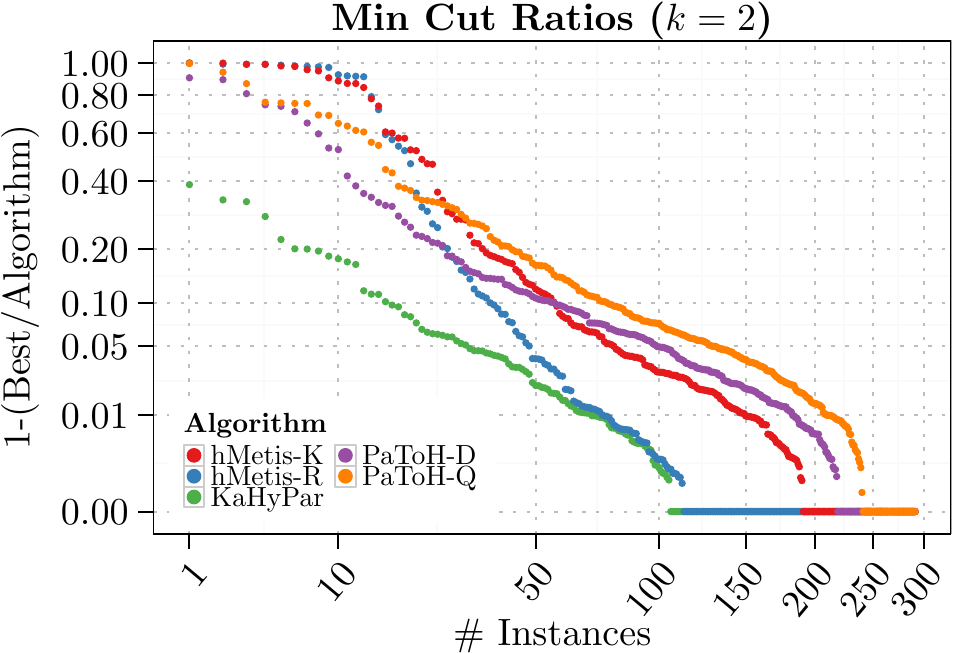} 
\includegraphics[width=.45\textwidth]{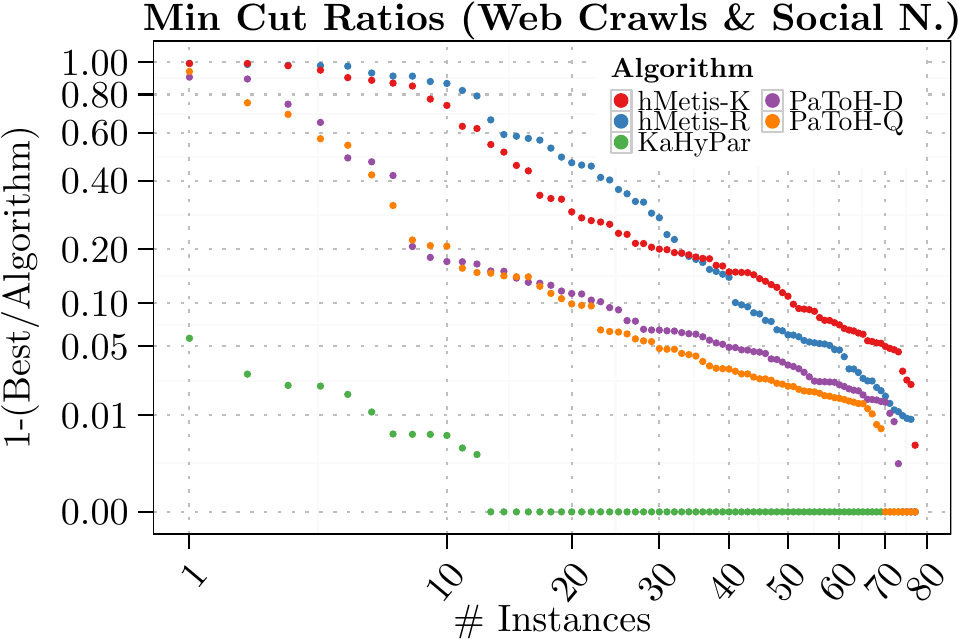} 

}

\end{knitrout}
\caption{Performance plots for subgroups of instances with large quality improvements ($\varepsilon = 0.03$): bipartitioning (left), web crawls and social networks (right). Note the cube root scale for both axes.}\label{fig:subgroup_eval}
\end{figure*}

\begin{figure*}[Ht!]
\centering
\begin{knitrout}
\definecolor{shadecolor}{rgb}{0.969, 0.969, 0.969}\color{fgcolor}

{\centering \includegraphics[width=.45\textwidth]{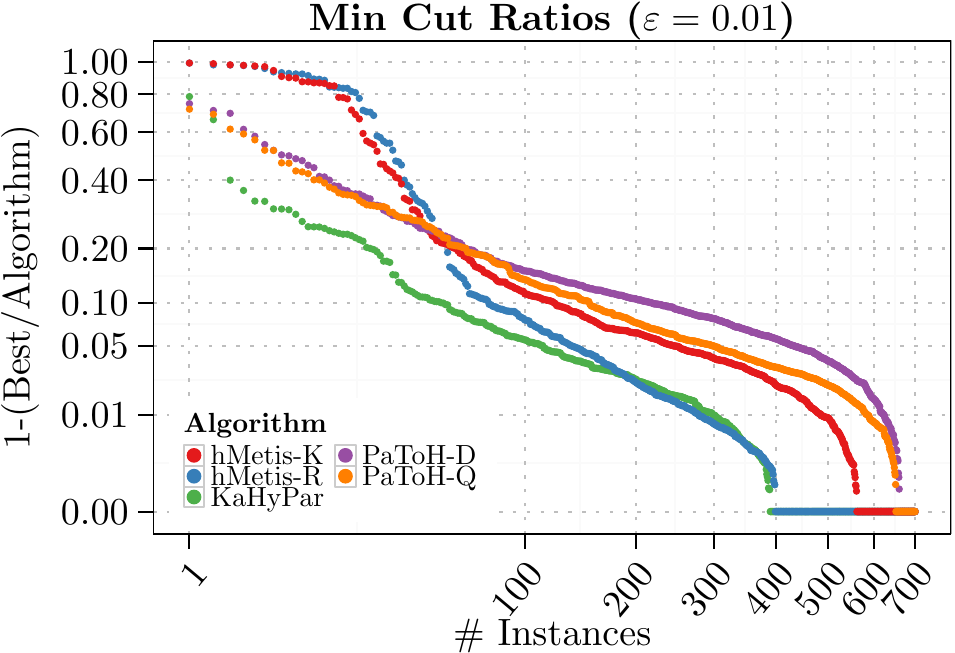} 
\includegraphics[width=.45\textwidth]{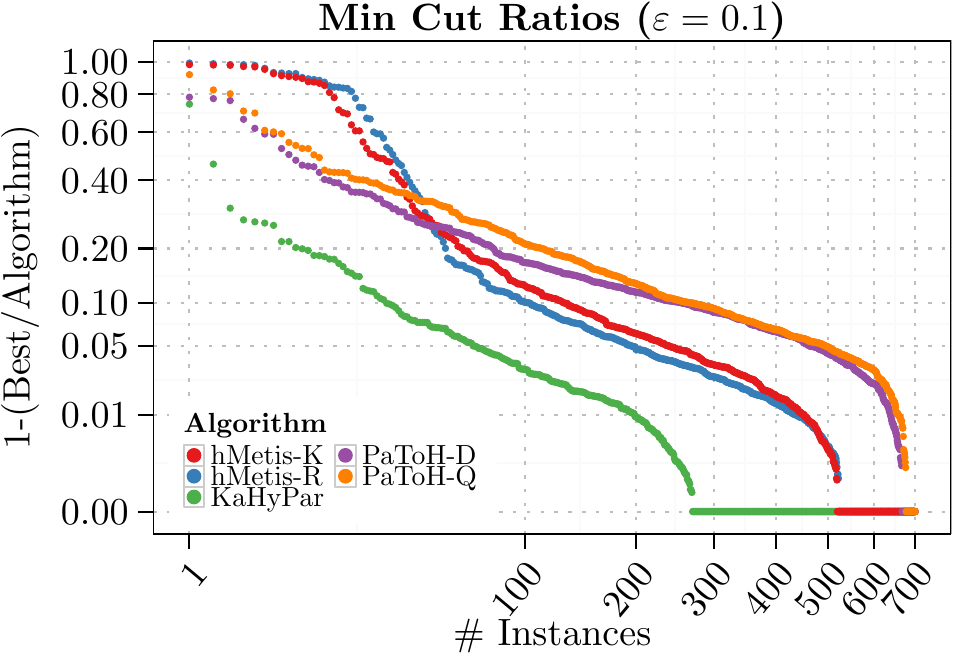} 

}

\end{knitrout}
\caption{Performance plots for $\varepsilon \in \{0.01,0.1\}$. Note the cube root scale for both axes.}\label{fig:epsilon_overview}
\end{figure*}

For certain well defined subgroups of instances significantly larger quality improvements are observed. \autoref{fig:subgroup_eval} show that this is the case for plain bipartitioning 
and for matrices derived from web crawls and social networks\footnote{Based on the following matrices: \texttt{webbase-1M}, \texttt{ca-CondMat}, \texttt{soc-sign-epinions}, \texttt{wb-edu}, 
\texttt{PGPgiantcompo}, \texttt{NotreDame\_www}, \texttt{NotreDame\_actors}, \texttt{IMDB}, \texttt{p2p-Gnutella25}, \texttt{Stanford}, \texttt{cnr-2000}}. 
On the other hand, for SAT instances, KaHyPar is slightly worse than hMetis-R. However, it is almost twice as fast in this case. On sparse matrix instances both hMetis variants
and both PaToH variants are inferior to KaHyPar regarding partitioning quality, while the running time of KaHyPar is comparable to hMetis. 
On VLSI hypergraphs, the performance in terms of running time and solution quality is comparable to hMetis-R.
Tables summarizing the average and best cuts found as well as the running times can be found in Appendix~\ref{app:geom}. Note that these tables also support
the interpretation of the results presented in this section.

\begin{figure}
\centering
\begin{knitrout}
\definecolor{shadecolor}{rgb}{0.969, 0.969, 0.969}\color{fgcolor}

{\centering \includegraphics[width=.42\textwidth]{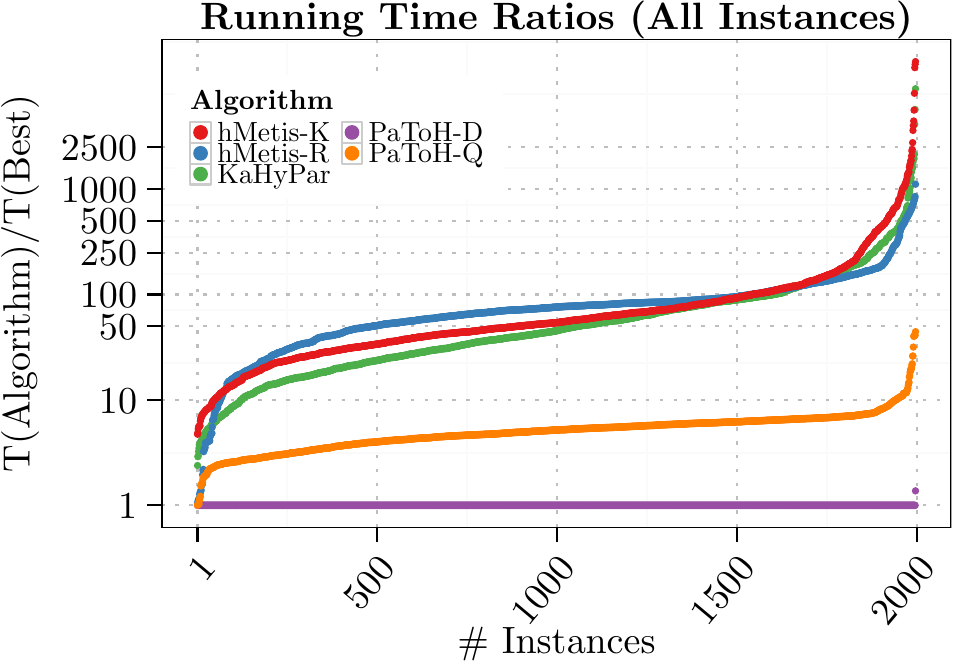} 
\includegraphics[width=.42\textwidth]{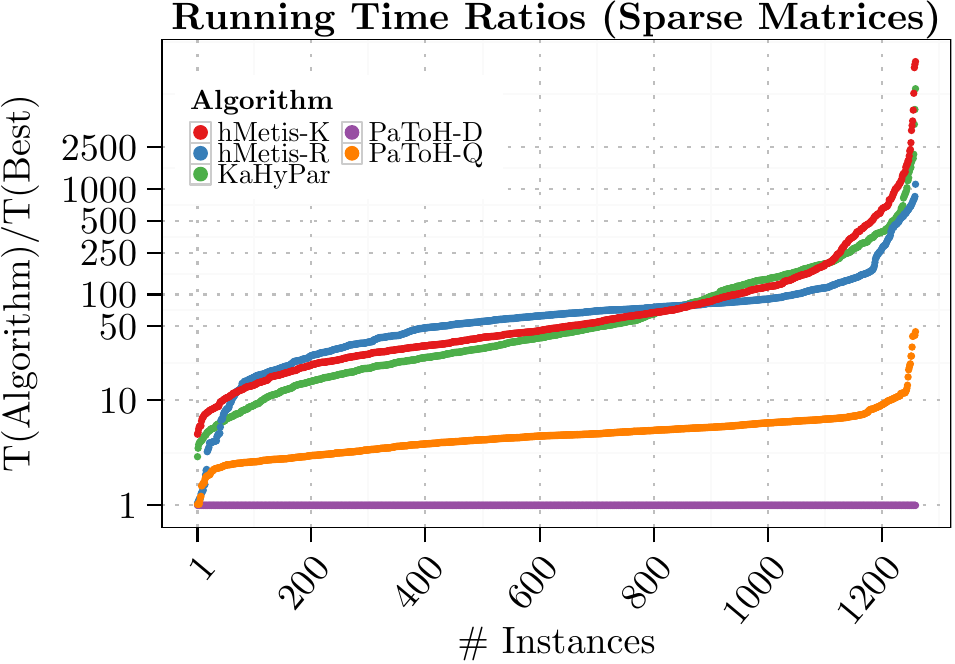} 
\includegraphics[width=.42\textwidth]{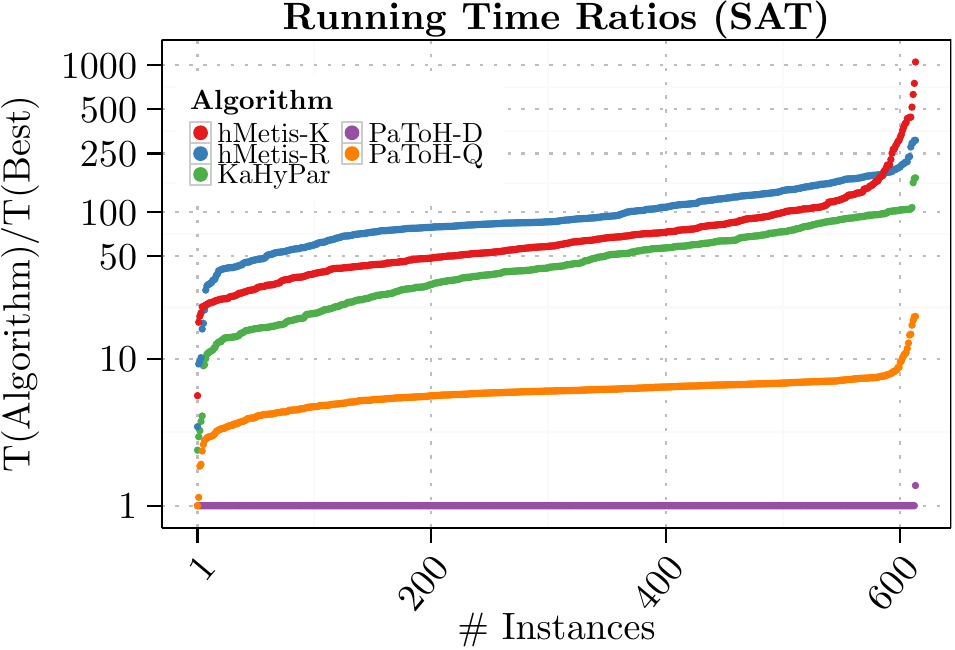} 
\includegraphics[width=.42\textwidth]{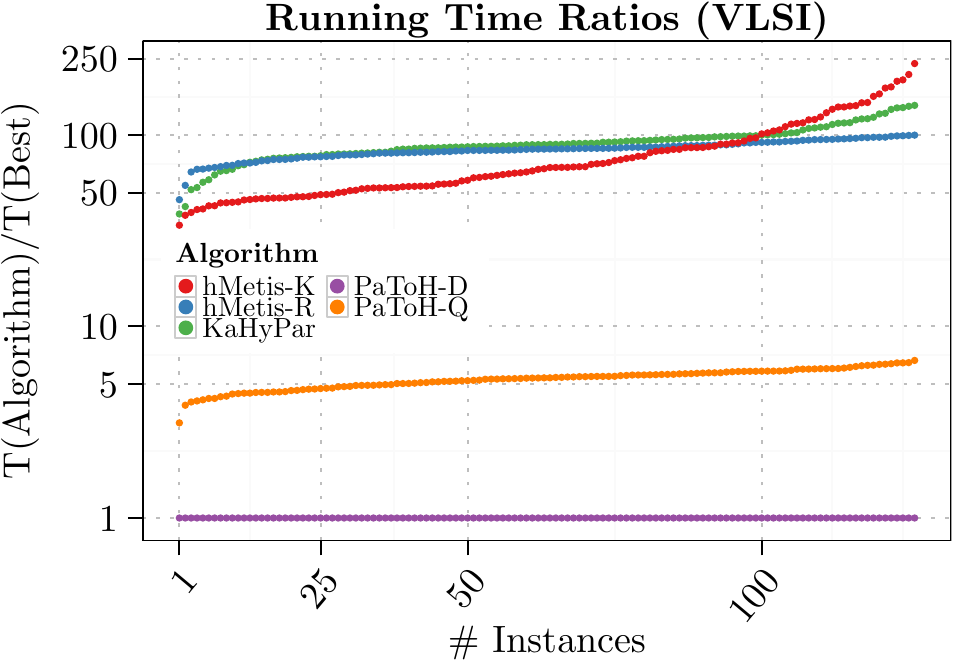} 

}

\end{knitrout}
\caption{Running time ratio plots for $\varepsilon=0.03$. The log-scaled y-axis shows the ratio between the running time of the corresponding algorithm and the running time of the fastest algorithm. For each algorithm, the ratios are sorted in increasing order.}\label{fig:running_time_overview}
\end{figure}

We conducted a Wilcoxon matched pairs signed rank test~\cite{wilcoxon} (using a 1\% significance level) to determine whether or not the difference of KaHyPar and the other algorithms is statistically significant. 
At a 1\% significance level, a z-score with $|Z| > 2.58$ is considered significant. We report the z-scores and the corresponding p-values.
Unless stated otherwise, the p-value reported was $ p < 2.2\cdot10^{-16}$.
For sparse matrices, the best cuts of KaHyPar are significantly better than those of hMetis-R ($Z=-13.861$),
hMetis-K ($Z=-17.953$), PaToH-Q ($Z=-26.894$) and PatoH-D ($Z=-26.273$). The difference of KaHyPar and hMetis-R was found to be not 
significant on SAT ($Z = 1.9943, p = 0.04612$) and VLSI instances ($Z = -0.4821, p = 0.6297$). Comparing the best solutions of KaHyPar
and hMetis-K, the difference was significant for SAT instances ($Z = -3.0483, p = 0.002302$) but not for VLSI instances ($Z = 0.18261, p = 0.8551$).
However, the solutions of KaHyPar were significantly better than PatoH-Q ($Z=-15.636$) and PaToH-D ($Z=-16.34$) on SAT instances
and VLSI instances (PatoH-Q: $Z=-9.5529$ and PaToH-D: $Z=-9.6125$).

\subsection{Effects of Imbalance Parameter.}\label{sec:different_imbalance}
For $\varepsilon=0.01$, $3087$ out of $7000$ partitions produced by hMetis-K were imbalanced (up to $14\%$ imbalance). 
hMetis-R produced 163 imbalanced partitions (up to $1.3\%$ imbalance), KaHyPar produced 58 imbalanced partitions (up to $2.5\%$ imbalance) 
and 205 partitions of PaToH-Q had up to $2\%$ imbalance. Only PaToH-D fulfilled the balance constraint in all experiments. 
For $\varepsilon=0.1$ hMetis-K produced 24 imbalanced partitions with up to $14\%$ imbalance. All other partitioners produced balanced partitions.
The overall performance is summarized in \autoref{fig:epsilon_overview}. Performance plots for each instance class can be found in \autoref{app:performanceplots}.
For an allowed imbalance of 1\% (left) the minimum cut ratios of KaHyPar are slightly better than hMetis-R. Due to the localized view of our local search algorithm it 
becomes difficult to find feasible moves around the uncontracted vertex pair that improve the solution quality.
However, our algorithm still outperforms hMetis-K, PaToH-D and PaToH-Q. If we allow up to 10\% imbalance (right), the performance of KaHyPar 
improves significantly. It now produces the best partitions for $431$ out of $700$ instances.
%Looking at the average and minimum cuts in~\autoref{tbl:epsilon001_results}, we see that if the balance constraint is tight ($\varepsilon = 0.01$), 
%KaHyPar performs worse than hMetis on SAT and VLSI instances. This can be explained by the fact that the average net size of these instances is small. 
%Due to the localized view of our local search algorithm it becomes difficult to find feasible moves around the uncontracted vertex pair that improve the solution quality.
%For hypergraphs with larger average net sizes (SPM instances), KaHyPar outperforms hMetis and is still slightly better than PaToH-Q. 

\section{Conclusions and Future Work}  \label{Conclusions}
We presented a novel $k$-way hypergraph partitioning algorithm that instantiates the multilevel paradigm in its most extreme version, by
removing only a single vertex between two levels. Our algorithm %is implemented in the hypergraph partitioning framework KaHyPar and 
computes high quality partitions for a large set of hypergraphs derived from various application domains. 
Four key aspects yield a tool/algorithm that dominates the popular hMetis system in \emph{both} solution quality and running time: (i) active exploitation of the fact that we represent the hypergraph as a bipartite graph to derive an efficient hypergraph data structure, (ii) an engineered coarsening algorithm, (iii) a portfolio of initial partitioning algorithms and (iv) a highly tuned local search algorithm.

Several ideas exist to narrow the gap between the running time of KaHyPar and PaToH: Ignoring nets that are larger than a certain
threshold size during coarsening, as suggested by~\cite{PaToH}. The running time of local search could be improved by developing an adaptive stopping rule as in~\cite{nGP}.
% Ignoring high degree vertices could further
%speed up local search, since moving these vertices is unlikely to improve the solution quality, but comes at the high cost of having to update the gains
%of all neighbors.
With respect to quality we could introduce  V-cycles~\cite{hMetisRB,SanSch11,WalshawVcycle}
and an evolutionary algorithm along the lines of KaHIP~\cite{SanSch12}.
Generalizing the gain cache to direct $k$-way partitioning 
as in \cite{DirectKwayNhgp} might give good quality for large $k$ without incurring excessive performance penalties.
Having shown that our algorithm computes high quality partitions when optimizing the total cut size,  future work could also look at
different partitioning objectives that rely on a global view of the problem, like the $(\lambda - 1)$ or sum-of-external-degrees metric~\cite{hMetisKway}.
%In order to optimize the former using recursive bisection, it suffices to switch from cut net removal to cut net splitting when constructing
%the section hypergraphs induced by a bipartition~\cite{PaToHManual}.

{
\bibliographystyle{unsrt}
\bibliography{library,refs-parco,diss}
}

\onecolumn
\appendix
\clearpage
\section{Parameter Tuning Instances} \label{app:tuninginstances}

\begin{table*}[h!]
\centering
   \caption{Properties of our benchmark set used for parameter tuning. The table is split into two groups: VLSI instances and sparse matrix instances. Within each group, the hypergraphs are sorted by $|V|$.}
\label{tbl:mediuminstances}
\begin{tabular}{m{1.7cm}|rrr|rrr|rrr}
\multirow{2}{*}{Hypergraph} & \multicolumn{1}{c}{\multirow{2}{*}{$|V|$}} & \multicolumn{1}{c}{\multirow{2}{*}{$|E|$}}  & \multicolumn{1}{c|}{\multirow{2}{*}{$|pins|$}}  & \multicolumn{3}{c|}{$d(v)$} & \multicolumn{3}{c}{$|e|$} \\
 & & & & \multicolumn{1}{c}{min} & \multicolumn{1}{c}{avg} & \multicolumn{1}{c|}{max} & \multicolumn{1}{c}{min} & \multicolumn{1}{c}{avg}& \multicolumn{1}{c}{max} \\
\cline{1-10}
ibm01 &  \numprint{12752} &  \numprint{14111} &  \numprint{ 50566} & 1 & 3.97 &   39 & 2 & 3.58 &  42 \\
ibm02 &  \numprint{19601} &  \numprint{19584} &  \numprint{ 81199} & 1 & 4.14 &   69 & 2 & 4.15 & 134 \\
ibm03 &  \numprint{23136} &  \numprint{27401} &  \numprint{ 93573} & 1 & 4.04 &  100 & 2 & 3.41 &  55 \\
ibm04 &  \numprint{27507} &  \numprint{31970} &  \numprint{105859} & 1 & 3.85 &  526 & 2 & 3.31 &  46 \\
ibm05 &  \numprint{29347} &  \numprint{28446} &  \numprint{126308} & 1 & 4.30 &    9 & 2 & 4.44 &  17 \\
ibm06 &  \numprint{32498} &  \numprint{34826} &  \numprint{128182} & 1 & 3.94 &   91 & 2 & 3.68 &  35 \\
ibm07 &  \numprint{45926} &  \numprint{48117} &  \numprint{175639} & 1 & 3.82 &   98 & 2 & 3.65 &  25 \\
\hline
vibrobox &  \numprint{12328} &  \numprint{12328} &  \numprint{ 342828} &  9 & 27.81 & 121 &  9 & 27.81 & 121 \\
bcsstk29 &  \numprint{13992} &  \numprint{13992} &  \numprint{ 619488} &  5 & 44.27 &  71 &  5 & 44.27 &  71 \\
memplus  &  \numprint{17758} &  \numprint{17758} &  \numprint{ 126150} &  2 &  7.10 & 574 &  2 &  7.10 & 574 \\
bcsstk30 &  \numprint{28924} &  \numprint{28924} &  \numprint{2043492} &  4 & 70.65 & 219 &  4 & 70.65 & 219 \\
bcsstk31 &  \numprint{35588} &  \numprint{35588} &  \numprint{1181416} &  2 & 33.20 & 189 &  2 & 33.20 & 189 \\
bcsstk32 &  \numprint{44609} &  \numprint{44609} &  \numprint{2014701} &  2 & 45.16 & 216 &  2 & 45.16 & 216 \\
\end{tabular}
\end{table*}

\section{Parameter Tuning Results} \label{app:parametertuning}
\begin{figure}[h!]
\centering
\begin{minipage}{.5\textwidth}
\vspace{0pt}
  \centering
  \includegraphics[width=.95\textwidth]{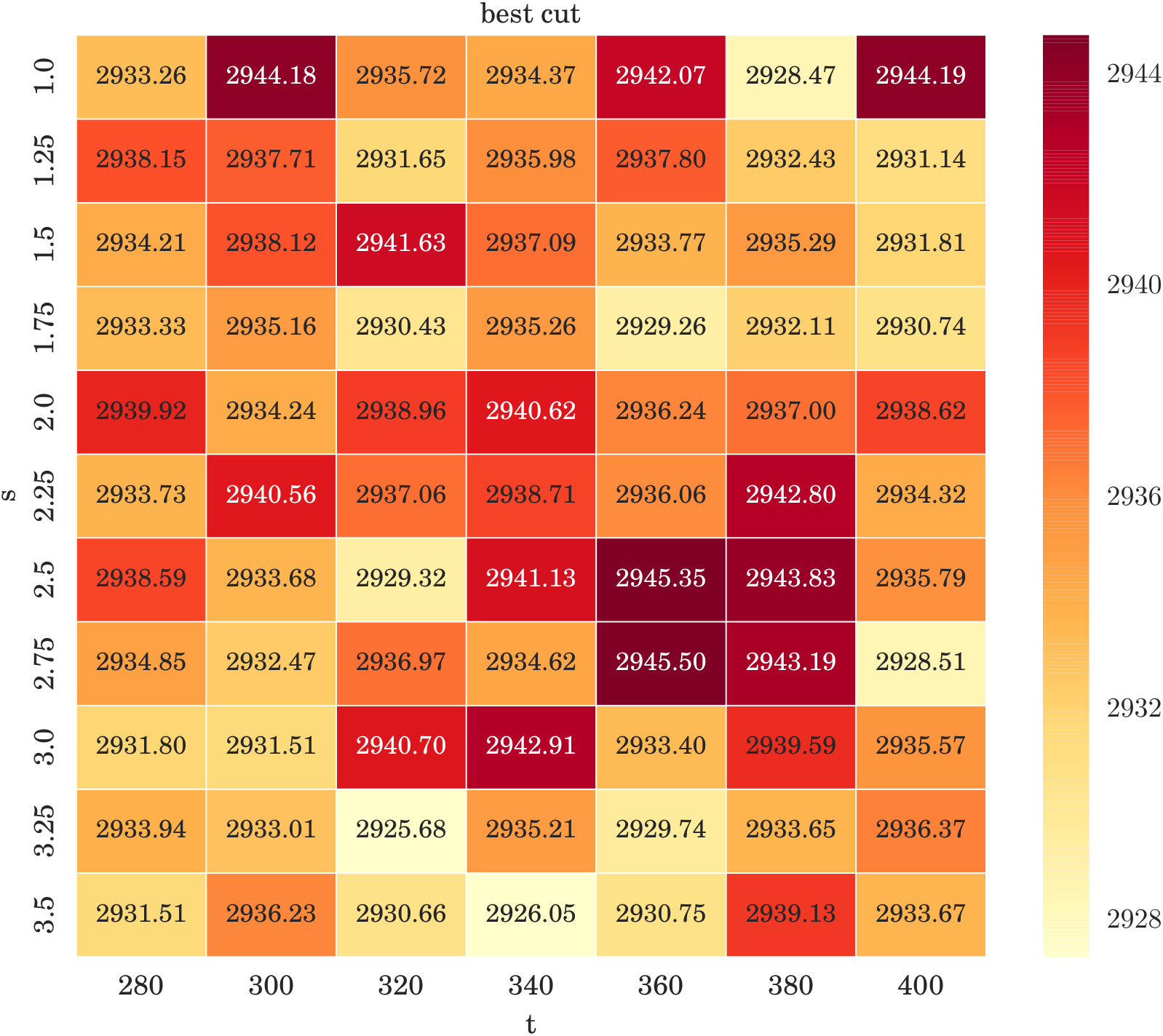}
\end{minipage}%
\begin{minipage}{.5\textwidth}
\vspace{-4pt}
  \centering
  \includegraphics[width=.95\textwidth]{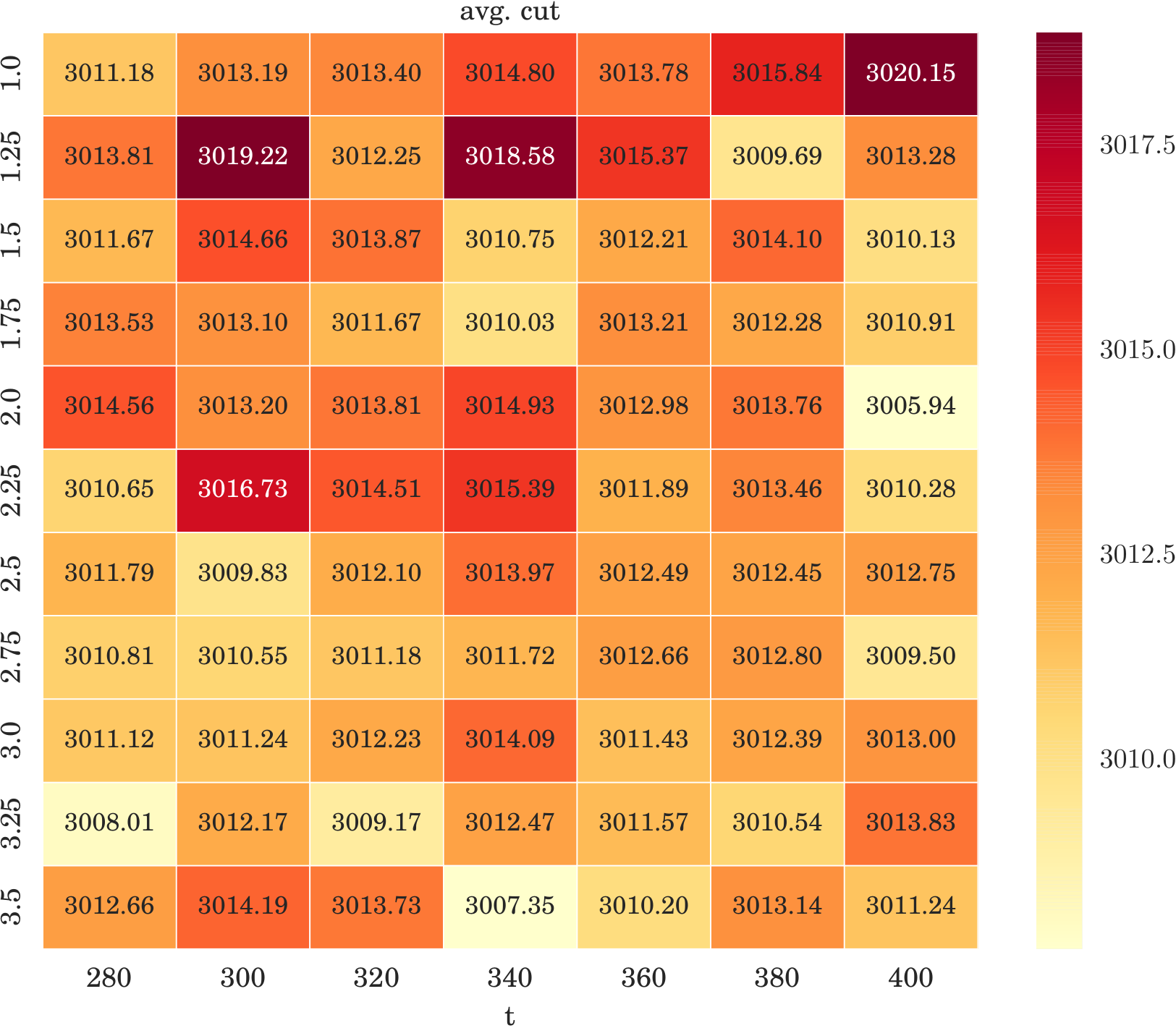}
\end{minipage}
\caption{Results of our parameter tuning experiments regarding the minimum size of the coarsest hypergraph $t$ and the scaling factor of maximum allowed vertex weight $s$. $t$-values
less than $280$ have been omitted because all of them resulted in worse solution quality. We set $s$ to $3.25$ and $t$ to $320$.} \label{tbl:sttuning}
\end{figure}

\begin{table}[t!]
\centering
\caption{Results of our parameter tuning experiments regarding the allowed maximum number of moves $c$ without improvement. Based on these experiments we set $c$ to $350$.}
\label{tbl:ctuning}
\begin{tabular}{rrrr}
   \multicolumn{1}{c}{$c$} & \multicolumn{1}{c}{avg. cut} & \multicolumn{1}{c}{best cut} & \multicolumn{1}{c}{t [$s$]} \\
\hline                                                      
  \numprint{10}            & \numprint{3044.8}            & \numprint{2969.7}            & \numprint{0.7}              \\
  \numprint{25}            & \numprint{3031.6}            & \numprint{2959.7}            & \numprint{1.3}              \\
  \numprint{50}            & \numprint{3026.1}            & \numprint{2959.2}            & \numprint{2.2}              \\
 \numprint{100}            & \numprint{3021.3}            & \numprint{2957.0}            & \numprint{3.9}              \\
 \numprint{150}            & \numprint{3020.4}            & \numprint{2954.7}            & \numprint{5.6}              \\
 \numprint{200}            & \numprint{3018.4}            & \numprint{2946.1}            & \numprint{7.3}              \\
 \numprint{250}            & \numprint{3017.4}            & \numprint{2945.3}            & \numprint{9.0}              \\
 \numprint{300}            & \numprint{3016.8}            & \numprint{2942.2}            & \numprint{10.9}             \\
 \textbf{\numprint{350}}   & \textbf{\numprint{3015.9}}   & \textbf{\numprint{2942.1}}   & \textbf{\numprint{12.7}}    \\
 \numprint{400}            & \numprint{3015.8}            & \numprint{2942.1}            & \numprint{14.9}             \\
 \numprint{450}            & \numprint{3015.2}            & \numprint{2943.9}            & \numprint{16.3}             \\
 \numprint{500}            & \numprint{3015.3}            & \numprint{2942.4}            & \numprint{17.8}             \\
\end{tabular}
\end{table}
\vspace{3cm}

\clearpage
\section{Excluded Test Instances} \label{app:excludedinstnaces}
\begin{table*}[h!]
\centering
\caption{Out of $2170$ test instances, we excluded the following $173$ instances either because PaToH-Q could not allocate enough memory or one of the other partitioners could
not partition the instances in the given time limit. The table is split into two groups: sparse matrix instances and SAT instances.}
\label{tbl:excluded}
\begin{tabular}{L{5cm}lllllll}
                                  & \multicolumn{7}{c}{$k$}                                                                                                                                                                                                                                                                                                 \\
Hypergraph                        & $2$                                           & $4$                                       & $8$                                       & $16$                                      & $32$                                      & $64$                                      & $128$                                     \\
\hline
12month1                          & $\triangle$                                   & $\triangle$                               & $\triangle$                               & $\triangle$                               & $\triangle$ \ding{109}                    & $\triangle$ \ding{109}                    & $\triangle$ \ding{109}                    \\
192bit                            & $\square$                                     &                                           &                                           &                                           &                                           &                                           &                                           \\
ASIC\_680k                        & $\triangle$                                   & $\triangle$                               & $\triangle$                               & $\triangle$                               & $\triangle$                               & $\triangle$                               & $\triangle$                               \\
ESOC                              & $\square$                                     & $\square$                                 &                                           &                                           & $\square$                                 & $\square$ \ding{109}                      & $\square$ \ding{109}                      \\
GL7d22                            & $\square$$\triangle$\ding{108} \ding{109}     & $\square$$\triangle$\ding{108} \ding{109} & $\square$$\triangle$\ding{108} \ding{109} & $\square$$\triangle$\ding{108} \ding{109} & $\square$$\triangle$\ding{108} \ding{109} & $\square$$\triangle$\ding{108} \ding{109} & $\square$$\triangle$\ding{108} \ding{109} \\
LargeRegFile                      & $\triangle$                                   & $\triangle$                               & $\triangle$                               & $\triangle$                               & $\triangle$                               & $\triangle$                               & $\triangle$                               \\ 
Rucci1                            &                                               &                                           &                                           &                                           & $\square$                                 &                                           &                                           \\
Trec14                            &                                               &                                           &                                           &                                           &                                           & \ding{109}                                & \ding{109}                                \\
appu                              &                                               &                                           &                                           &                                           &                                           & \ding{109}                                & \ding{109}                                \\
circuit5M                         & $\square$$\triangle$                          & $\square$$\triangle$                      & $\square$$\triangle$\ding{108}            & $\square$$\triangle$\ding{108}            & $\square$$\triangle$\ding{108} \ding{109} & $\square$$\triangle$\ding{108} \ding{109} & $\square$$\triangle$\ding{108} \ding{109} \\
gupta3                            &                                               &                                           &                                           &                                           &                                           &                                           & \ding{109}                                \\
hollywood-2009                    & $\triangle$                                   & $\triangle$\ding{108} \ding{109}          & $\triangle$\ding{108} \ding{109}          & $\triangle$\ding{108} \ding{109}          & $\triangle$\ding{108} \ding{109}          & $\triangle$\ding{108} \ding{109}          & $\triangle$\ding{108} \ding{109}          \\
human\_gene2                      &                                               &                                           &                                           &                                           & \ding{109}                                & \ding{109}                                & \ding{109}                                \\
kron\_g500-logn16                 &                                               &                                           &                                           &                                           & \ding{109}                                & \ding{109}                                & \ding{109}                                \\
n4c6-b11                          & $\square$                                     & $\square$                                 & $\square$                                 & $\square$                                 & $\square$                                 & $\square$                                 & $\square$                                 \\
nd12k                             &                                               &                                           &                                           &                                           &                                           &                                           & \ding{109}                                \\
nlpkkt120                         &                                               &                                           &                                           & \ding{108}                                & \ding{108}                                & \ding{108}                                & \ding{108}                                \\
rel9                              & $\square$                                     & $\square$$\triangle$\ding{108} \ding{109} & $\square$$\triangle$\ding{108} \ding{109} & $\square$$\triangle$\ding{108} \ding{109} & $\square$$\triangle$\ding{108} \ding{109} & $\square$$\triangle$\ding{108} \ding{109} & $\square$$\triangle$\ding{108} \ding{109} \\
sls                               & $\square$                                     & $\square$                                 & $\square$ \ding{109}                      & $\square$ \ding{109}                      & $\square$ \ding{109}                      & $\square$ \ding{109}                      & $\square$ \ding{109}                      \\
\hline                                                                                                                                                                                                                                                                                                                  
004-80-8                          & $\square$                                     & $\square$                                 & $\square$                                 & $\square$                                 & $\square$                                 & $\square$                                 & $\square$                              \\
005-80-12                         & $\square$                                     & $\square$                                 & $\square$                                 & $\square$                                 & $\square$                                 & $\square$                                 & $\square$                              \\
007-80-8                          & $\square$                                     & $\square$                                 & $\square$                                 & $\square$                                 & $\square$                                 & $\square$                                 & $\square$                              \\
008-80-12                         & $\square$                                     & $\square$                                 & $\square$                                 & $\square$                                 & $\square$                                 & $\square$                                 & $\square$                              \\
008-80-8                          & $\square$                                     & $\square$                                 & $\square$                                 & $\square$                                 & $\square$                                 & $\square$                                 & $\square$                              \\
010-80-12                         & $\square$                                     & $\square$                                 & $\square$                                 & $\square$                                 & $\square$                                 & $\square$                                 & $\square$                              \\
11pipe\_k                         &                                               &                                           & \ding{109}                                & \ding{109}                                & \ding{109}                                & \ding{109}                                & \ding{109}                             \\
11pipe\_q0\_k                     &                                               &                                           &                                           &                                           &                                           &                                           & \ding{109}                             \\
9vliw\_m\_9stages \_iq3\_C1\_bug10 & \ding{109}                                    & \ding{108} \ding{109}                     & \ding{108} \ding{109}                     & \ding{108} \ding{109}                     & $\triangle$\ding{108} \ding{109}          & $\triangle$\ding{108} \ding{109}          & $\triangle$\ding{108} \ding{109}               \\
9vliw\_m\_9stages \_iq3\_C1\_bug7  &                                               & \ding{108} \ding{109}                     & \ding{108} \ding{109}                     & $\triangle$\ding{108} \ding{109}          & $\triangle$\ding{108} \ding{109}          & $\triangle$\ding{108} \ding{109}          & $\triangle$\ding{108} \ding{109}               \\
9vliw\_m\_9stages \_iq3\_C1\_bug8  &                                               & \ding{108} \ding{109}                     & \ding{108} \ding{109}                     & \ding{108} \ding{109}                     & $\triangle$\ding{108} \ding{109}          & $\triangle$\ding{108} \ding{109}          & $\triangle$\ding{108} \ding{109}               \\
bjrb07amba 10andenv                & $\square$                                     & $\square$                                 & $\square$                                 & $\square$                                 & $\square$                                 & $\square$                                 & $\square$                              \\
blocks-blocks-37-1. 130-NOTKNOWN   &                                               & $\square$                                 & $\square$                                 & $\square$                                 & $\square$                                 & $\square$                                 & $\square$                              \\
q\_query\_3\_L150\_coli.sat       &                                               &                                           &                                           &                                           &                                           &                                           & \ding{109}                             \\
q\_query\_3\_L200\_coli.sat       &                                               &                                           & \ding{109}                                & \ding{109}                                & \ding{109}                                & $\triangle$ \ding{109}                    & $\triangle$ \ding{109}                          \\
velev-vliw-uns-2.0-uq5            &                                               &                                           &                                           &                                           &                                           &                                           & \ding{109}                             \\
\end{tabular}
\caption*{
    \begin{tabular}{|l l|}
      \hline
      $\triangle$~:               & KaHyPar exceeded time limit \\
      \ding{108}~:                & hMetis-R exceeded time limit \\
      \ding{109}~:                & hMetis-K exceeded time limit \\
      $\square$~:                 & PaToH-Q memory allocation error \\
      \hline
    \end{tabular}
  }
\end{table*}

\clearpage
\section{Geometric Mean Comparison}\label{app:geom}
\begin{table*}[h!]
  \centering
  \caption{Comparison with other systems for all 1997 test instances and $\varepsilon=0.03$. Cuts of hMetis and PaToH are shown as increase in \emph{percent} relative to KaHyPar.}  \label{tbl:allresults}
  \begin{tabular}{L{1.5cm}|SSR{.6cm}|SSR{.6cm}|SSR{.6cm}|SSR{.6cm}}
$\varepsilon=0.03$ & \multicolumn{3}{c|}{All Instances}           & \multicolumn{3}{c|}{Sparse Matrices}         & \multicolumn{3}{c|}{SAT}           & \multicolumn{3}{c}{VLSI}                                                                                                                                                                                                                                                                                                                                                     \\
Algorithm          & \multicolumn{1}{c}{avg.}                     & \multicolumn{1}{c}{best}                     & \multicolumn{1}{R{.6cm}|}{t [$s$]} & \multicolumn{1}{c}{avg.}                     & \multicolumn{1}{c}{best}                     & \multicolumn{1}{R{.6cm}|}{t [$s$]} & \multicolumn{1}{c}{avg.}             & \multicolumn{1}{c}{best}             & \multicolumn{1}{R{.6cm}|}{t [$s$]} & \multicolumn{1}{c}{avg.}            & \multicolumn{1}{c}{best}                     & \multicolumn{1}{R{.6cm}}{t [$s$]} \\
\hline                                                                                                                                                                                       
KaHyPar            & \multicolumn{1}{r}{\textbf{\numprint{6776}}} & \multicolumn{1}{r}{\textbf{\numprint{6581}}} & \numprint{37}                      & \multicolumn{1}{r}{\textbf{\numprint{4801}}} & \multicolumn{1}{r}{\textbf{\numprint{4696}}} & \numprint{30}                      & \multicolumn{1}{r}{\numprint{14821}} & \multicolumn{1}{r}{\numprint{14177}} & \numprint{59}                      & \multicolumn{1}{r}{\numprint{4694}} & \multicolumn{1}{r}{\textbf{\numprint{4575}}} & \numprint{29}                     \\
hMetis-R           & +7.91                                        & +8.80                                        & \numprint{51}                      & +14.52                                       & +14.53                                       & \numprint{35}                      & \multicolumn{1}{r}{\textbf{-2.75}}   & \multicolumn{1}{r}{\textbf{ -0.38}}  & \numprint{131}                     & \multicolumn{1}{r}{\textbf{-1.15}}  & +0.09                                        & \numprint{27}                     \\ 
hMetis-K           & +10.19                                       & +10.06                                       & \numprint{47}                      & +14.55                                       & +14.74                                       & \numprint{36}                      & +3.94                                & +2.97                                & \numprint{92}                      & -0.68                               & +0.41                                        & \numprint{23} \\
PaToH-Q            & +4.39                                        & \multicolumn{1}{c}{--}                       & \numprint{4}                       & +3.56                                        & \multicolumn{1}{c}{--}                       & \numprint{3}                       & +6.33                                & \multicolumn{1}{c}{--}               & \numprint{8}                       & +3.34                               & \multicolumn{1}{c}{--}                       & \numprint{2}  \\                                                                
PaToH-D            & +10.05                                       & +7.81                                        & \numprint{1}                       & +8.20                                        & +6.15                                        & \numprint{1}                       & +14.02                               & +11.47                               & \numprint{1}                       & +9.69                               & +6.94                                        & \numprint{1}  \\ 
\end{tabular}

\vspace{.5cm}

\begin{minipage}{0.45\textwidth}
\centering
\caption{All Benchmark Instances for $k=2$ and $\varepsilon=0.03$} \label{tbl:bipartitioning}
\begin{tabular}{lrrr}
Algorithm & \multicolumn{1}{c}{avg.}   & \multicolumn{1}{c}{best}   & \multicolumn{1}{c}{t [$s$]} \\
\hline                                                                    
KaHyPar    & \textbf{\numprint{1173}}   & \textbf{\numprint{1105}}   & \numprint{11.8}             \\
hMetis-R   & \numprint{+19.96}      & \numprint{+21.64}      & \numprint{20.6}             \\
hMetis-K   & \numprint{+22.72}      & \numprint{+25.13}      & \numprint{19.2}             \\
PaToH-Q    & \numprint{+7.97}       & \multicolumn{1}{c}{--}     & \numprint{1.5}              \\
PaToH-D    & \numprint{+11.48}      & \numprint{+8.19}       & \numprint{0.3}              \\
\end{tabular}
\end{minipage}
%\vspace{1em}
\begin{minipage}{0.45\textwidth}
\centering
\caption{Web-Crawls and Social Networks for $\varepsilon=0.03$} \label{tbl:spmwebsocial}
\begin{tabular}{lrrr}
Algorithm & \multicolumn{1}{c}{avg.} & \multicolumn{1}{c}{best} & \multicolumn{1}{c}{t [$s$]} \\
\hline                                                                    
KaHyPar    & \textbf{\numprint{6680}} & \textbf{\numprint{6269}} & \numprint{190}            \\
hMetis-R   & \numprint{+86.63}    & \numprint{+92.03}    & \numprint{198}           \\
hMetis-K   & \numprint{+66.82}    & \numprint{+71.64}    & \numprint{263}           \\
PaToH-Q    & \numprint{+ 9.03}    & \multicolumn{1}{c}{--}   & \numprint{7}             \\
PaToH-D    & \numprint{+17.88}    & \numprint{+18.72}    & \numprint{1}              \\
\end{tabular}
\end{minipage}

\vspace{.5cm}

  \caption{Comparison with other systems for 700 test instances and $\varepsilon=0.01$ (top) and  $\varepsilon=0.1$ (bottom).} \label{tbl:epsilon001_results}
  \begin{tabular}{L{1.5cm}|SSR{.6cm}|SSR{.6cm}|SSR{.6cm}|SSR{.6cm}}
$\varepsilon=0.01$ & \multicolumn{3}{c|}{All Instances}           & \multicolumn{3}{c|}{Sparse Matrices}         & \multicolumn{3}{c|}{SAT}           & \multicolumn{3}{c}{VLSI} \\
Algorithm          & \multicolumn{1}{c}{avg.}                     & \multicolumn{1}{c}{best}                     & \multicolumn{1}{R{.6cm}|}{t [$s$]} & \multicolumn{1}{c}{avg.}                     & \multicolumn{1}{c}{best}                     & \multicolumn{1}{R{.6cm}|}{t [$s$]} & \multicolumn{1}{c}{avg.}             & \multicolumn{1}{c}{best}             & \multicolumn{1}{R{.6cm}|}{t [$s$]} & \multicolumn{1}{c}{avg.}            & \multicolumn{1}{c}{best}            & \multicolumn{1}{R{.6cm}}{t [$s$]} \\
\hline                                                                                                                                                                                       
KaHyPar            & \multicolumn{1}{r}{\textbf{\numprint{6272}}} & \multicolumn{1}{r}{\textbf{\numprint{6065}}} & \numprint{37}                      & \multicolumn{1}{r}{\textbf{\numprint{4386}}} & \multicolumn{1}{r}{\textbf{\numprint{4280}}} & \numprint{29}                      & \multicolumn{1}{r}{\numprint{12367}} & \multicolumn{1}{r}{\numprint{11725}} & \numprint{54}                      & \multicolumn{1}{r}{\numprint{6992}} & \multicolumn{1}{r}{\numprint{6793}} & \numprint{51} \\
hMetis-R           & +8.77                                        & +9.74                                        & \numprint{60}                      & +19.44                                       & +19.46                                       & \numprint{41}                      & \multicolumn{1}{r}{\textbf{-6.43}}   & \multicolumn{1}{r}{\textbf{ -4.06}}  & \numprint{136}                     & \multicolumn{1}{r}{\textbf{-2.56}}  & \multicolumn{1}{r}{\textbf{-1.28}}  & \numprint{52} \\ 
hMetis-K           & +10.63                                       & +10.92                                       & \numprint{44}                      & +19.81                                       & +19.26                                       & \numprint{33}                      & -2.26                                & -0.73                                & \numprint{82}                      & -0.62                               & +0.11                               & \numprint{38}                     \\
PaToH-Q            & +1.77                                        & \multicolumn{1}{c}{--}                       & \numprint{4}                       & +1.73                                        & \multicolumn{1}{c}{--}                       & \numprint{3}                       & +1.69                                & \multicolumn{1}{c}{--}               & \numprint{8}                       & +2.30                               & \multicolumn{1}{c}{--}              & \numprint{3}                      \\                                                                
PaToH-D            & +10.44                                       & +7.35                                        & \numprint{1}                       & +9.07                                        & +6.41                                        & \numprint{1}                       & +13.06                               & +9.06                                 & \numprint{1}                       & +10.94                              & +7.88                               & \numprint{1}  \\ 
\end{tabular}

\vspace{.5cm}

\begin{tabular}{L{1.5cm}|SSR{.6cm}|SSR{.6cm}|SSR{.6cm}|SSR{.6cm}}
$\varepsilon=0.1$ & \multicolumn{3}{c|}{All Instances}           & \multicolumn{3}{c|}{Sparse Matrices}         & \multicolumn{3}{c|}{SAT}           & \multicolumn{3}{c}{VLSI}                                                                                                                                                                                                                                                                                                                                                     \\
Algorithm         & \multicolumn{1}{c}{avg.}                     & \multicolumn{1}{c}{best}                     & \multicolumn{1}{R{.6cm}|}{t [$s$]} & \multicolumn{1}{c}{avg.}                     & \multicolumn{1}{c}{best}                     & \multicolumn{1}{R{.6cm}|}{t [$s$]} & \multicolumn{1}{c}{avg.}             & \multicolumn{1}{c}{best}                      & \multicolumn{1}{R{.6cm}|}{t [$s$]} & \multicolumn{1}{c}{avg.}            & \multicolumn{1}{c}{best}            & \multicolumn{1}{R{.6cm}}{t [$s$]} \\
\hline                                                                                                                                                                                       
KaHyPar           & \multicolumn{1}{r}{\textbf{\numprint{5847}}} & \multicolumn{1}{r}{\textbf{\numprint{5676}}} & \numprint{36}                      & \multicolumn{1}{r}{\textbf{\numprint{4154}}} & \multicolumn{1}{r}{\textbf{\numprint{4061}}} & \numprint{28}                      & \multicolumn{1}{r}{\numprint{11132}} & \multicolumn{1}{r}{\textbf{\numprint{10626}}} & \numprint{52}                      & \multicolumn{1}{r}{\numprint{6597}} & \multicolumn{1}{r}{\numprint{6455}} & \numprint{49}                     \\
hMetis-R          & +11.74                                       & +12.60                                       & \numprint{58}                      & +21.79                                       & +21.56                                       & \numprint{40}                      & \multicolumn{1}{r}{\textbf{-2.50}}   & +0.20                                         & \numprint{132}                     & +0.34                               & +0.96                               & \numprint{50}                     \\ 
hMetis-K          & +11.33                                       & +11.74                                       & \numprint{44}                      & +20.12                                       & +20.21                                       & \numprint{34}                      & -0.62                                & +0.37                                         & \numprint{81}                      & \multicolumn{1}{r}{\textbf{-0.78}}  & \multicolumn{1}{r}{\textbf{-0.54}}  & \numprint{38}                     \\
PaToH-Q           & +9.16                                        & \multicolumn{1}{c}{--}                       & \numprint{4}                       & +7.41                                        & \multicolumn{1}{c}{--}                       & \numprint{3}                       & +12.98                               & \multicolumn{1}{c}{--}                        & \numprint{8}                       & +8.43                               & \multicolumn{1}{c}{--}              & \numprint{3}                      \\                                                                
PaToH-D           & +13.02                                       & +10.15                                       & \numprint{1}                       & +11.20                                       & +8.58                                        & \numprint{1}                       & +16.29                               & +13.27                                         & \numprint{1}                       & +14.37                              & +10.34                              & \numprint{1}  \\ 
\end{tabular}
\end{table*}

\clearpage
\section{Comparison for Different Imbalance Values} \label{app:performanceplots}
\begin{figure*}[h!]
\centering
\begin{knitrout}
\definecolor{shadecolor}{rgb}{0.969, 0.969, 0.969}\color{fgcolor}

{\centering \includegraphics[width=.45\textwidth]{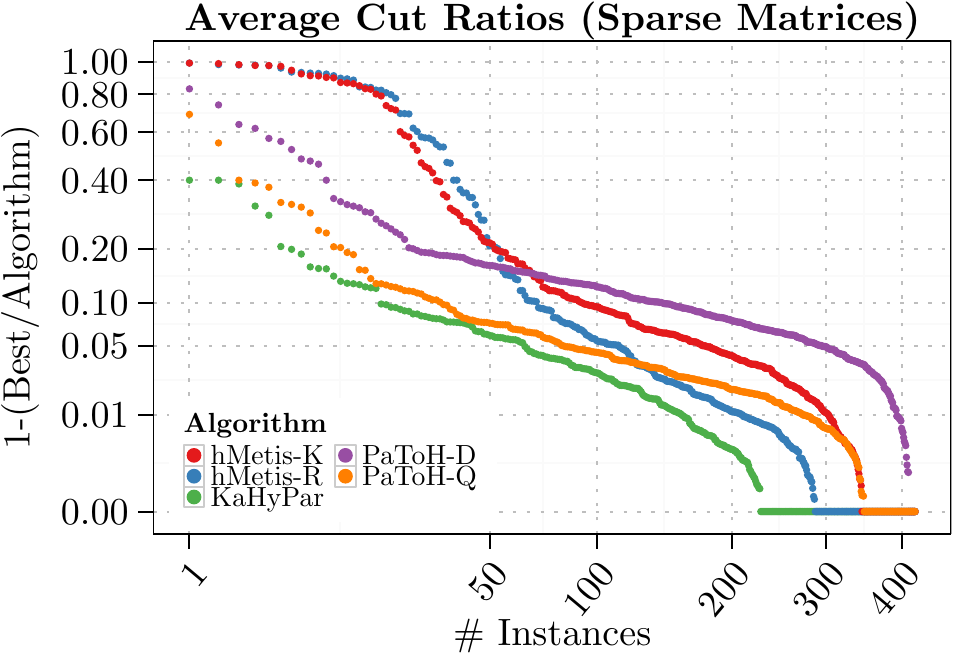} 
\includegraphics[width=.45\textwidth]{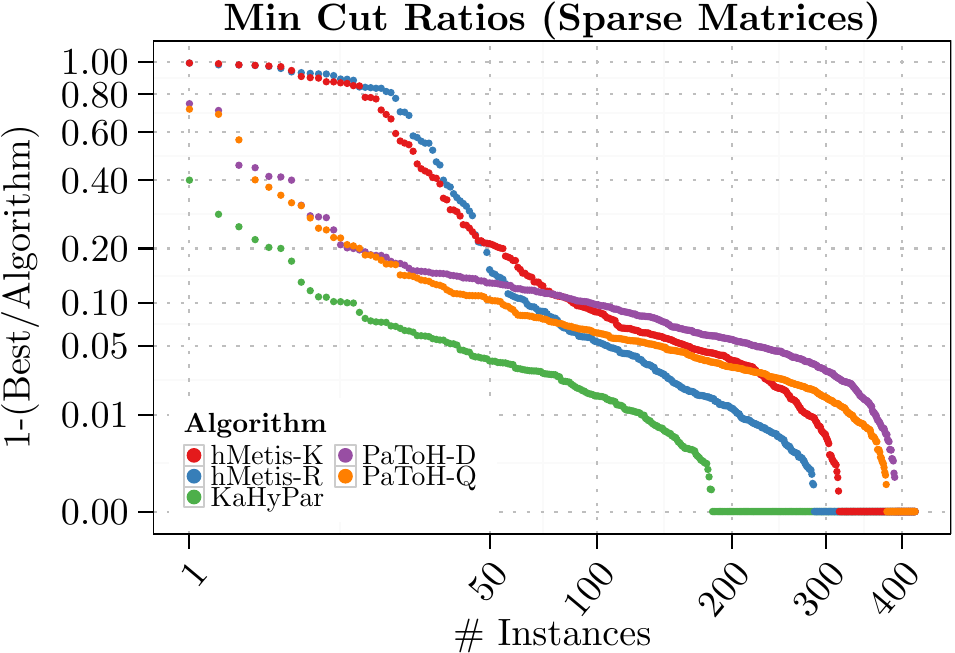} 
\includegraphics[width=.45\textwidth]{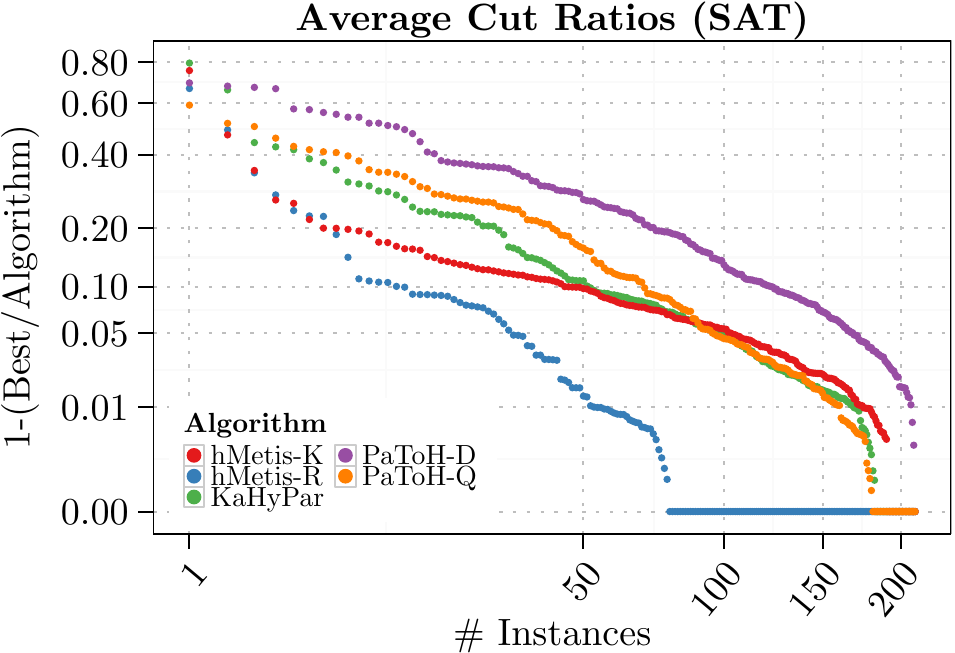} 
\includegraphics[width=.45\textwidth]{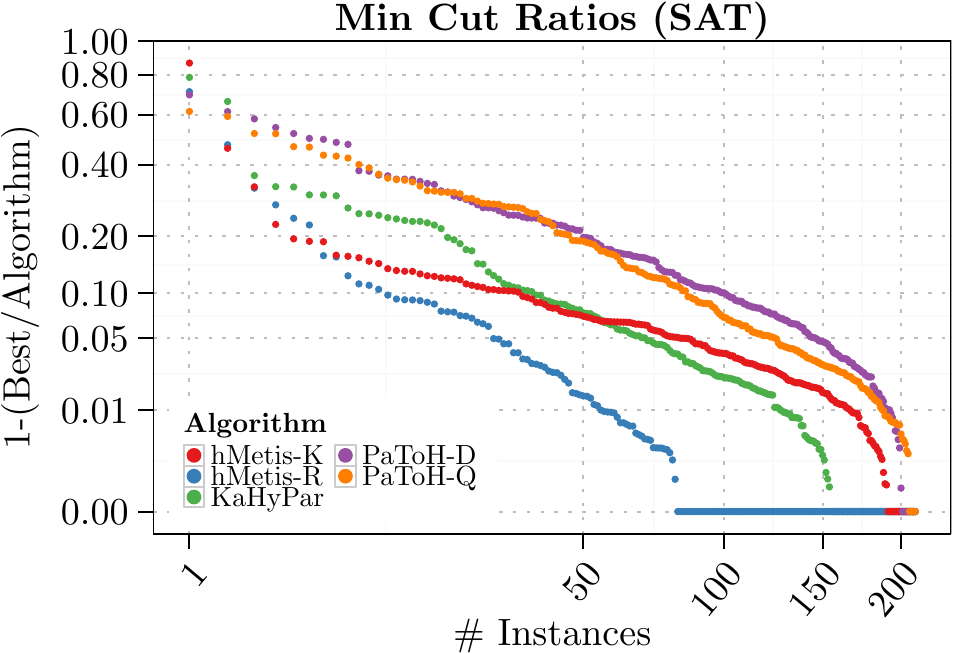} 
\includegraphics[width=.45\textwidth]{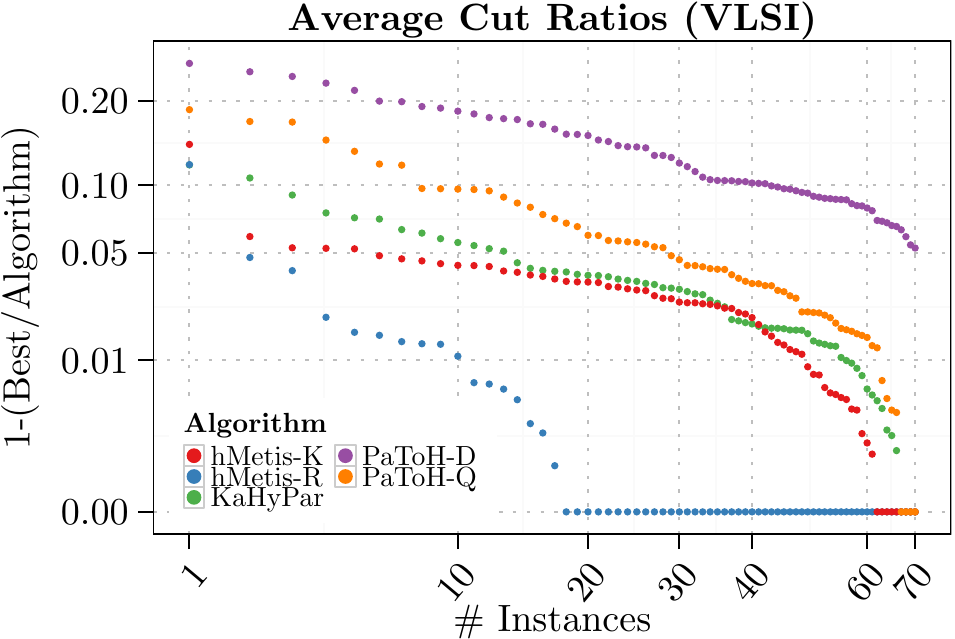} 
\includegraphics[width=.45\textwidth]{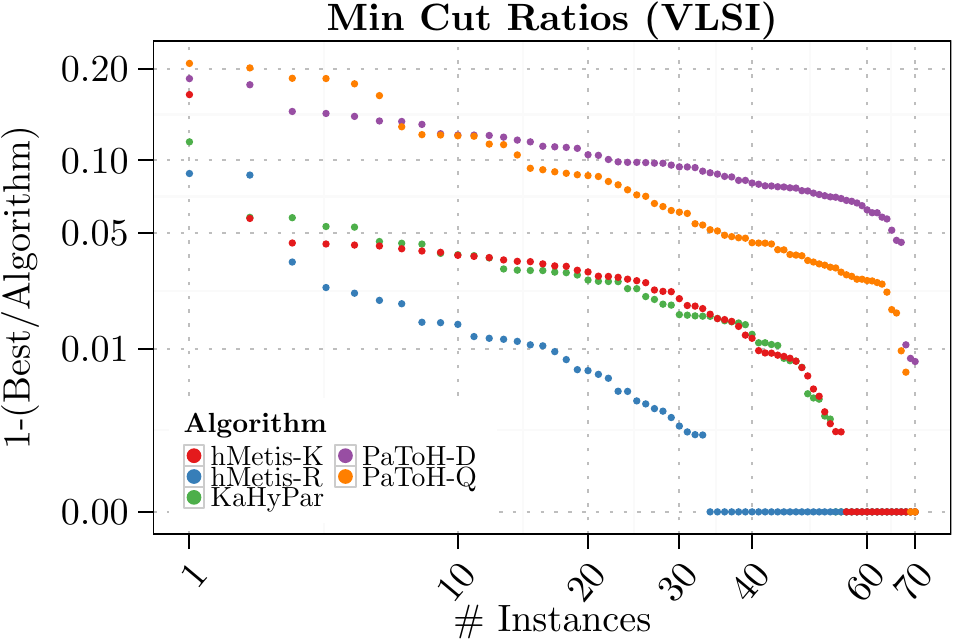} 

}
\end{knitrout}
\caption{Performance plots per benchmark set for $\varepsilon=0.01$. Note the cube root scale for both axes.}\label{fig:epsilon_001}
\end{figure*}

\begin{figure*}[h!]
\centering
\begin{knitrout}
\definecolor{shadecolor}{rgb}{0.969, 0.969, 0.969}\color{fgcolor}

{\centering \includegraphics[width=.45\textwidth]{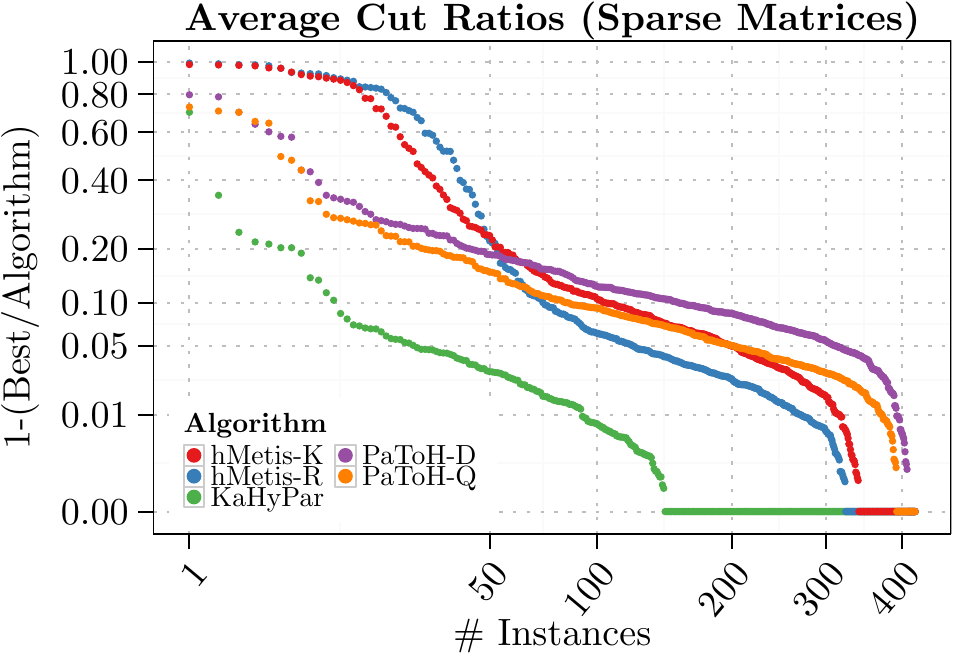} 
\includegraphics[width=.45\textwidth]{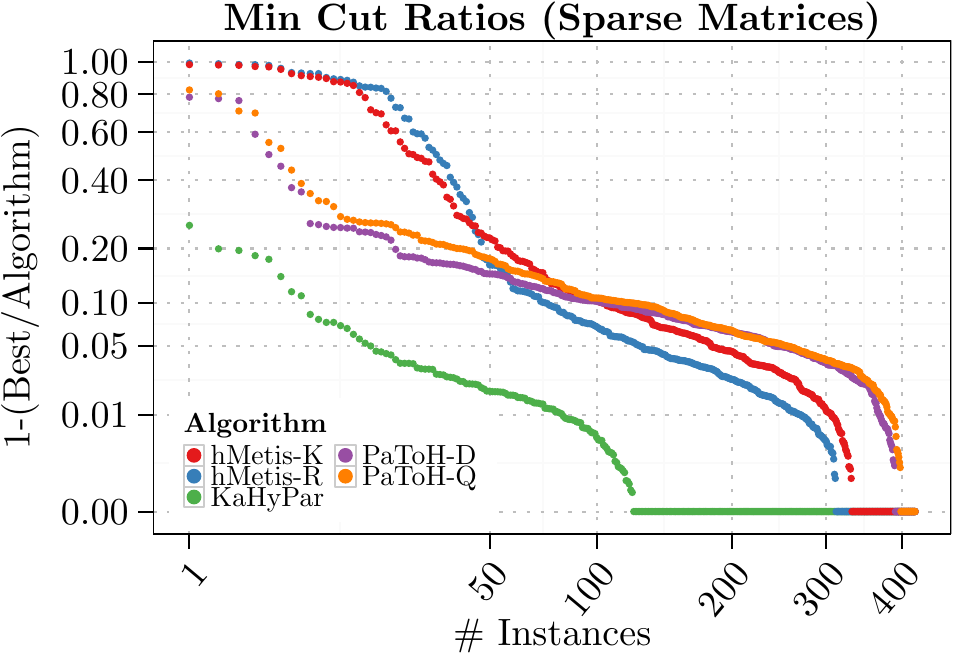} 
\includegraphics[width=.45\textwidth]{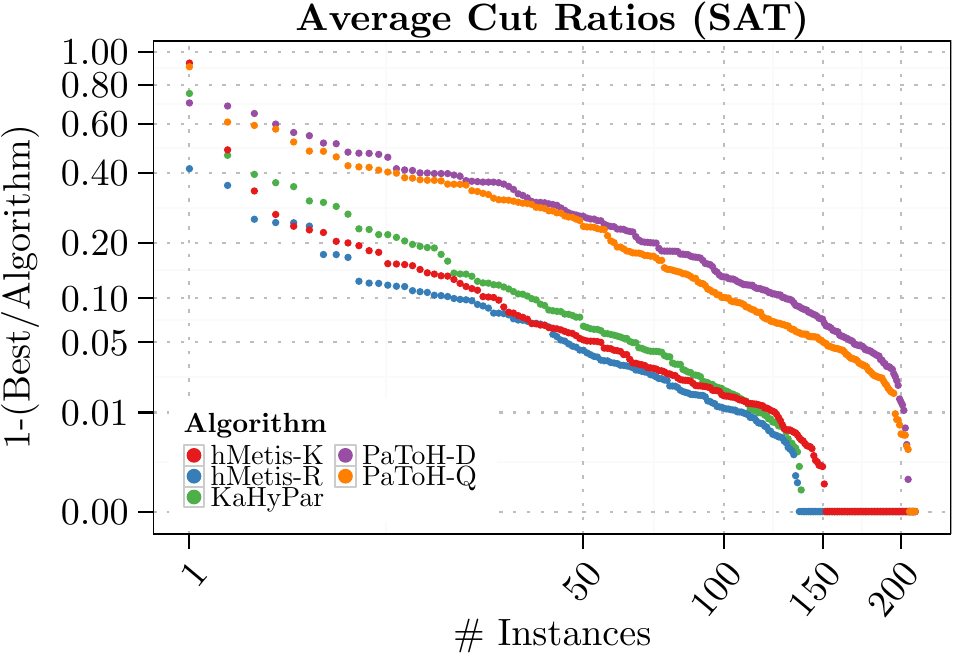} 
\includegraphics[width=.45\textwidth]{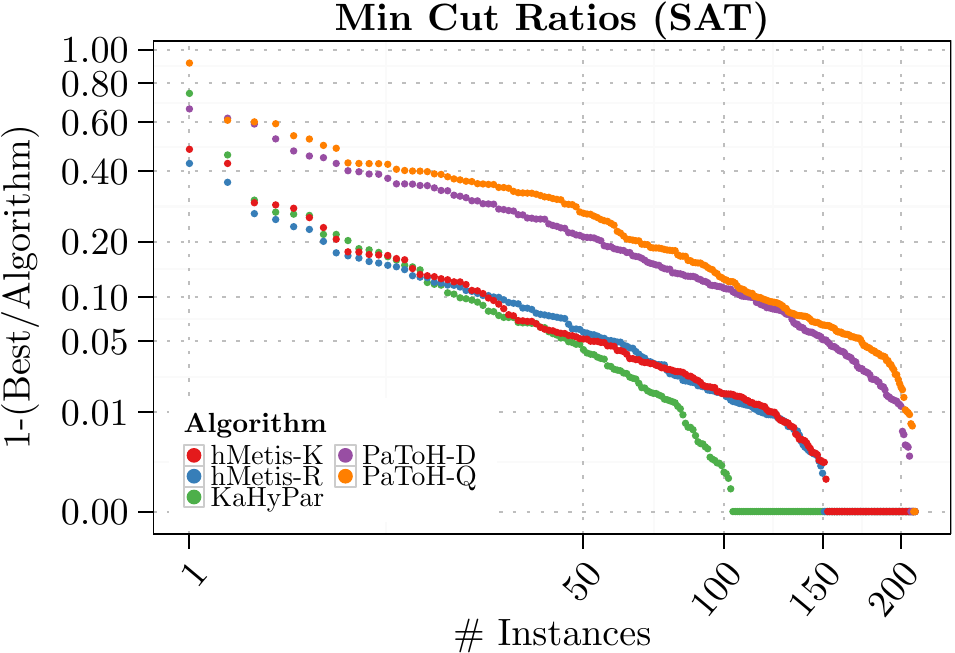} 
\includegraphics[width=.45\textwidth]{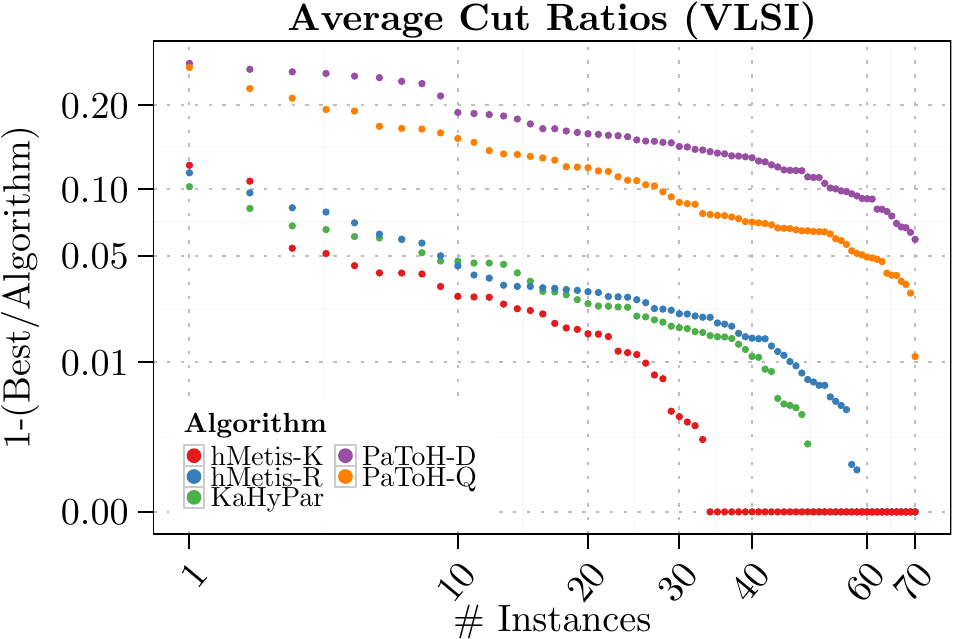} 
\includegraphics[width=.45\textwidth]{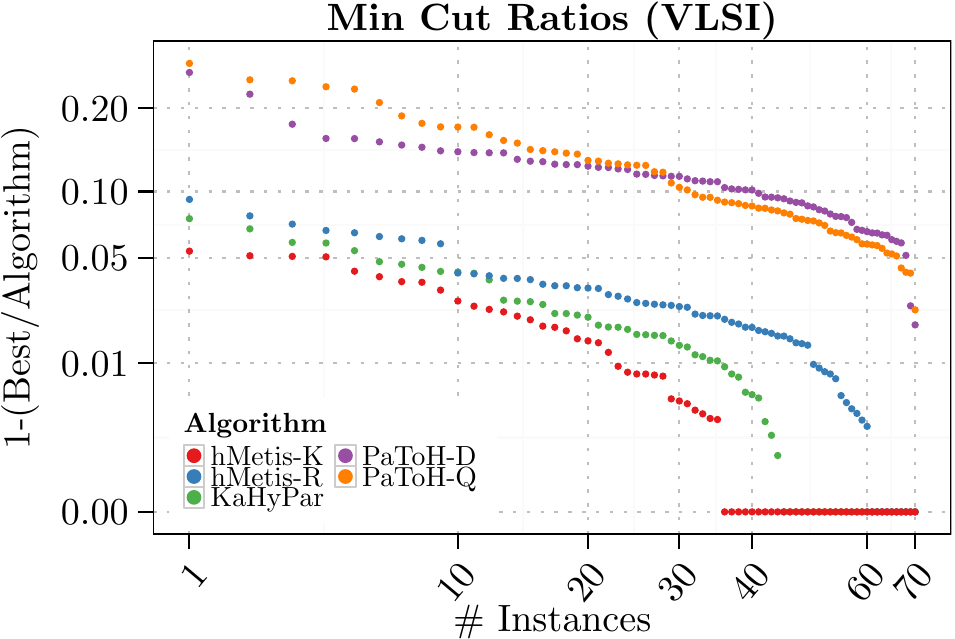} 

}

\end{knitrout}
\caption{Performance plots per benchmark set for $\varepsilon=0.1$. Note the cube root scale for both axes.}\label{fig:epsilon_01}
\end{figure*}

\end{document}